\pdfoutput=1
\RequirePackage{ifpdf}
\ifpdf 
\documentclass[pdftex]{sigma}
\else
\documentclass{sigma}
\fi

\numberwithin{equation}{section}

\newtheorem{thm}{Theorem}[section]
\newtheorem*{thm*}{Theorem}
\newtheorem{cor}[thm]{Corollary}
\newtheorem{lem}[thm]{Lemma}
\newtheorem{prop}[thm]{Proposition}
{ \theoremstyle{definition}
	\newtheorem{dfn}[thm]{Definition}
	
	\newtheorem{eg}[thm]{Example}
	\newtheorem{rmk}[thm]{Remark} }

\usepackage{tikz,tikz-cd}
\usepackage{slashed}
\usepackage{mathrsfs}

\def\cD{\mathcal D}
\def\cE{\mathcal E}
\def\cK{\mathcal K}
\def\cN{\mathcal N}

\def\CC{\mathbb C}

\def\RR{\mathbb R}

\def\ZZ{\mathbb Z}

\def\sL{\mathscr L}

\def\sR{\mathscr R}

\def\fg{\mathfrak g}\def\fh{\mathfrak h}  \def\fn{\mathfrak n}\def\fp{\mathfrak p}

\def\cat#1{\ensuremath{\mathsf{#1}}}

\DeclareMathOperator{\Mod}{\cat{Mod}}
\DeclareMathOperator{\Mult}{\cat{Mult}}

\def\lie#1{\ensuremath{\mathfrak{#1}}}
\def\aut{\ensuremath{\lie{aut}}}
\def\aff{\ensuremath{\lie{aff}}}
\def\so{\lie{so}}

\def\bu{\bullet}

\def\d{{\rm d}}

\DeclareMathOperator{\Spec}{Spec}
\DeclareMathOperator{\Spin}{Spin}
\DeclareMathOperator{\End}{End}
\DeclareMathOperator{\Sym}{Sym}
\DeclareMathOperator{\id}{id}
\DeclareMathOperator{\Hom}{Hom}
\DeclareMathOperator{\Vect}{Vect}
\DeclareMathOperator{\Ext}{Ext}

\def\sL{\mathscr L}
\def\sR{\mathscr R}

\begin{document}
\allowdisplaybreaks

\newcommand{\arXivNumber}{2205.14133}

\renewcommand{\PaperNumber}{022}

\FirstPageHeading

\ShortArticleName{The Derived Pure Spinor Formalism as an Equivalence of Categories}

\ArticleName{The Derived Pure Spinor Formalism\\ as an Equivalence of Categories}

\Author{Chris ELLIOTT~$^{\rm a}$, Fabian HAHNER~$^{\rm b}$ and Ingmar SABERI~$^{\rm c}$}

\AuthorNameForHeading{C.~Elliott, F.~Hahner and I.~Saberi}

\Address{$^{\rm a)}$~Department of Mathematics and Statistics, Amherst College,\\
\hphantom{$^{\rm a)}$}~220 South Pleasant Street, Amherst, MA 01002, USA}
\EmailD{\href{mailto:celliott@amherst.edu}{celliott@amherst.edu}}

\Address{$^{\rm b)}$~Mathematisches Institut der Universit\"at Heidelberg, \\
	\hphantom{$^{\rm b)}$}~Im Neuenheimer Feld 205, 69120 Heidelberg, Germany}
\EmailD{\href{mailto:fhahner@mathi.uni-heidelberg.de}{fhahner@mathi.uni-heidelberg.de}}

\Address{$^{\rm c)}$~Ludwig-Maximilians-Universit\"at M\"unchen, \\
	\hphantom{$^{\rm c)}$} Theresienstra\ss{}e 37, 80333 M\"unchen, Germany}
\EmailD{\href{mailto:i.saberi@physik.uni-muenchen.de}{i.saberi@physik.uni-muenchen.de}}

\ArticleDates{Received July 12, 2022, in final form April 04, 2023; Published online April 18, 2023}

\Abstract{We construct a derived generalization of the pure spinor superfield formalism and prove that it exhibits an equivalence of dg-categories between multiplets for a supertranslation algebra and equivariant modules over its Chevalley--Eilenberg cochains. This equivalence is closely linked to Koszul duality for the supertranslation algebra. After introducing and describing the category of supermultiplets, we define the derived pure spinor construction explicitly as a dg-functor. We then show that the functor that takes the derived supertranslation invariants of any supermultiplet is a quasi-inverse to the pure spinor construction, using an explicit calculation. Finally, we illustrate our findings with examples and use insights from the derived formalism to answer some questions regarding the ordinary (underived) pure spinor superfield formalism.}

\Keywords{pure spinor superfields; equivalence of categories; supersymmetry; field theory; BV formalism}

\Classification{81R25; 17B55; 17B81}

\section{Introduction}

Much effort has been expended over the last fifty years to develop efficient techniques for constructions of, and computations in, supersymmetric field theories. The essential difficulty is visible even in the simplest models: because one asks that the supersymmetry transformations square to ordinary translations, the supersymmetry transformations on the space of fields typically only define a representation of the supersymmetry algebra after imposing equations of motion for fermionic fields. Similarly, if gauge fields are present, the algebra is only represented up to gauge transformations.
For the purposes of quantization, it is desirable for the symmetry to act on the full space of fields, without regard to the dynamics; to attain this, one must often pass to a more complicated model, which tends to involve additional ``auxiliary fields'' that do not change the physics of the theory. Such auxiliary fields may or may not exist, depending on the example in which one is interested.

Relatedly, since the action of the Poincar\'e group on spacetime is geometric in nature and the super Poincar\'e group is (under certain assumptions) the only interesting nontrivial possibility for extended spacetime symmetry that is available for consideration~\cite{HLS}, it is pleasing to think of supersymmetry as arising from the action of particular geometric symmetries on an appropriate supermanifold, which amounts to giving ``superfield'' formulations of supersymmetric theories.
Since supersymmetry necessarily acts on the space of superfields, giving such a formulation implicitly requires extending the field content of the theory, in one fashion or another, by some set of auxiliary fields as described above.
(The literature on supersymmetry is immense, and we cannot hope to give an overview here; for foundational work in the subject, the reader is referred to the collection of reprints~\cite{FerraraBook}, or the early review~\cite{Sohnius} and references therein.)

Finding systematic techniques to construct superfield formulations and extend the range of theories for which they are available has thus been a subject of great interest. Numerous approaches have been developed, including harmonic superspace~\cite{HSbook} and the ``rheonomy'' approach to supergravity theories~\cite{CdAF}. There have also been approaches in which the space of so-called \emph{pure spinors} is used to construct appropriate sets of auxiliary fields, beginning more than thirty years ago in papers by Nilsson~\cite{Nilsson} and Howe~\cite{HowePS1,HowePS2}. Pure spinors were used to great effect in Berkovits' formulation of the superstring~\cite{Berkovits}, and their applications to superfield formulations have been developed in a wide variety of examples, notably in work of Cederwall and collaborators. See~\cite{character, CNT2, CederwallM5} for early papers, and~\cite{Cederwall} for a review with references to further literature.

There have been numerous related studies of ideas connected to the space of pure spinors, for example by Krotov, Losev and collaborators~\cite{AKLL,KroLoCS} and, notably, in the work of Movshev and Schwarz on ten-dimensional supersymmetric Yang--Mills theory~\cite{Mov1,Mov2,MovshevGeometryPureSpinor, MS1,MS2}. See also the related work of Kapranov \cite{KapranovSuper}. This work was further developed in a mathematical context in~\cite{GKR}, in~\cite{GS}, and in~\cite{GGST}, where connections to the theory of Koszul duality were emphasized.
(In particular, the graded Lie algebra Koszul dual to functions on the nilpotence variety was studied, and used to compute the cohomology of the pure spinor superfield.)
Relatedly, work by Movshev, Schwarz, and Xu on the Lie algebra cohomology of supersymmetry algebras~\cite{MSX1,MSX2} made the appearance of the space of square-zero elements in that context clear. The cohomology of the supertranslation algebra appears in various guises in the physics literature, notably in the context of the ``brane scan'' \cite{OBS, FSS, HuertaThesis}, where certain cohomology classes are identified as WZW terms in worldvolume action functionals.

\begin{rmk}
	It is worth making an orienting remark on terminology at this point. The term ``pure spinor'' is often used in the physics literature to refer to points in (or coordinates on) the space of odd square-zero elements in the supertranslation algebra. This space is sometimes also called the nilpotence variety~\cite{NV}. However, the relevant condition here is that of being a~Maurer--Cartan element (having zero self-bracket), which has nothing at all to do with the spin representation, at least \emph{a priori}.
	On the other hand, pure spinors, in the sense of Cartan and Chevalley~\cite{ChevalleySpinor}, are defined to be elements of the spin representation for which the dimension of the annihilator under Clifford multiplication is maximal; the pure spinors form the minimal orbit in the action of the spin group on the (projectivized) spin representation.
	In ten-dimensional minimal supersymmetry, which was the first example studied~\cite{Nilsson}, the odd elements consist of a~single copy of the spin representation, the bracket is defined using Clifford multiplication, and the spaces of square-zero elements and pure spinors coincide. This coincidence is responsible for the confusing terminology; our usage of ``pure spinor formalism'' in this paper is historically, rather than logically, motivated.
\end{rmk}

The pure spinor superfield formalism was reinterpreted in terms of sheaves over the nilpotence variety in~\cite{NV}; a detailed analysis of the formalism in the general setting of appropriate generalizations of the supertranslation algebra---odd central extensions of abelian Lie algebras---was given in~\cite{perspectives}. In this version, the pure spinor superfield formalism constructs a supersymmetric multiplet (a collection of classical fields equipped with an action of the supersymmetry algebra) out of the datum of a graded module over the ring of functions on the nilpotence variety of the supertranslation algebra.
The output of the pure spinor superfield formalism is a locally free sheaf on superspace, equipped with a differential and a module structure making it a multiplet. One can think of it as resolving the more standard component field formulations freely over superspace (and thus a type of superfield formalism). There is a filtration that allows one to pass to a minimal quasi-isomorphic presentation that is a locally free sheaf on the spacetime; this minimal presentation agrees in examples with the typical component-field descriptions of supermultiplets in the literature.
Depending on the example being considered, a set of auxiliary fields may be generated automatically; more generally, though, the output includes differentials that resolve the equations of motion of the free multiplet, thus giving a description in the Batalin--Vilkovisky formalism. In this formalism, auxiliary fields are not essential, and an on-shell action of supersymmetry is formalized as a homotopy action of supersymmetry on the derived critical locus of the action.

It is natural to ask for a characterization of all multiplets which arise in this manner.
In~\cite{perspectives}, an example of a multiplet which cannot be constructed using the technique was described.\footnote{The obstruction to constructing this multiplet using the pure spinor formalism was previously observed by Martin Cederwall.} The failure in this particular case was linked to the geometry of the underlying nilpotence variety: it is not Cohen--Macaulay, so that the dualizing complex of its ring of functions has cohomology in multiple degrees and is not quasi-isomorphic to a single homogeneous module. At the level of multiplets, the structure sheaf gives rise to the four-dimensional $\cN = 1$ vector multiplet. On general grounds, one expects that the dualizing module gives rise to
the dual (or ``antifield'') multiplet.
However, due to the failure of the Cohen--Macaulay property, this does not work on the nose. Instead, the appearance of the dualizing complex suggests that a generalization of the pure spinor superfield formalism to the world of derived algebraic geometry is necessary.

In this paper, we tackle these questions systematically, and show that a derived generalization of the pure spinor formalism can be used to produce \emph{every} supermultiplet in a very general setting.
Let $\fn_2$ and $\fn_1$ be vector spaces. We regard $\fn_2$ and $\Pi \fn_1$ (where $\Pi$ denotes the parity shift functor) as abelian super Lie algebras, and consider an odd central extension of the form
\begin{equation} \label{eq:ext1}
	0 \to \fn_2 \to \fn \to \Pi\fn_1 \to 0.
\end{equation}
Since the Lie algebra cohomology group $\mathrm H^2(\Pi \fn_1)$ with trivial coefficients is, here, just given as $\Sym^2 \fn_1^*$, such a central extension is equivalent to a choice of map $\Gamma \colon \Sym^2(\fn_1) \to \fn_2$ that defines the bracket.
We can extend by a chosen subalgebra $\fg_0$ of the even automorphisms of~$\fn$, getting another extension sequence of the form
\begin{equation*}
	0 \to \fn \to \fg \to \fg_0 \to 0.
\end{equation*}
($\fn$ clearly has no even inner automorphisms.) The super Lie algebra $\fg$ will act on the supergroup $N = \exp(\fn)$, and its even subalgebra $\fg_+$ will act on the body $V$ of~$N$. The automorphism Lie algebra $\aut(\fn)$ will include an abelian summand than can be thought of as ``scale transformations,'' giving rise to an internal weight grading which places $\fn_1$ in degree one, $\fn_2$ in degree two, and $\fg_0 \to \aut(\fn)$ in degree zero; this will either be internal or not, depending on whether the summand is included in $\fg_0$.

Our motivating example will be
the case where $V=\RR^n$ is a Riemannian vector space, $\fg_+$ is the Poincar\'e algebra of infinitesimal isometries of $\RR^n$, and $\fg$ is the super Poincar\'e algebra. However, other examples are relevant to physics as well, for example when considering twisted theories in the pure spinor formalism along the lines of~\cite{spinortwist}.

Let $\Mult_{\fg}$ denote the dg-category of \emph{$\fg$-multiplets}: roughly speaking, these will be differential graded vector bundles on $V$, equipped with a $\fg$-action for which the action of~$\fg_+$ is ``geometric'', i.e., inherited from the action on spacetime. We will denote multiplets by triples $(E,D,\rho)$, consisting of a vector bundle, differential and $\fg$-action; precise definitions will be given in Section~\ref{sec: multiplets}. Our main result, stated somewhat informally, says the following.

\begin{thm*}[Theorem \ref{thm: equiv}]
	There is an equivalence of dg-categories
	\begin{equation}
	\label{informal_equivalence}
	\Mult_{\fg} \leftrightarrows \Mod_{C^\bullet(\fn)}^{\fg_0}
	\end{equation}
	between $\fg$-multiplets and $\fg_0$-equivariant modules over the Chevalley--Eilenberg algebra $C^\bullet(\fn)$.
\end{thm*}

We will view the functor from left to right as a natural derived enhancement of the pure spinor construction. Indeed, $C^\bullet(\fn)$ can be regarded as a bigraded commutative differential algebra, since the action of rescalings in $\fg_0$ has integral eigenvalues.
Taking Lie algebra cochains of the exact sequence~\eqref{eq:ext1}, we obtain the sequence
\begin{equation*}
	\CC \to C^\bu(\Pi \fn_1) \to C^\bu(\fn) \to C^\bu(\fn_2) \to \CC
\end{equation*}
of bigraded cdgas, witnessing $C^\bu(\fn)$ as an $R$-algebra, where $R := C^\bu(\Pi \fn_1)$ is the free commutative algebra on~$\fn_1^*$.
If we totalize the bigrading, the complex is concentrated in non-positive degrees. Its degree-zero cohomology then takes the form $R/I$, where $I$ is the quadratic ideal that arises as the image of the dual $\Gamma^*$ to the bracket map. The (affine) nilpotence variety is $\Spec(R/I)$, so that modules over this ring correspond to quasi-coherent sheaves on the nilpotence variety. Furthermore, there is a natural equivariant map
\begin{equation*}
	C^\bu(\fn) \to R/I,
\end{equation*}
so that any $R/I$-module is a $C^\bu(\fn)$-module in a natural way. When applied to modules of this sort, our derived enhancement agrees on the nose with the standard pure spinor formalism.

From this point of view, we can view $C^\bullet(\fn)$ as a derived enhancement of the nilpotence variety: its spectrum is an affine derived scheme whose underlying classical affine scheme is exactly the affine nilpotence variety, but whose derived data remembers the higher cohomology of $\fn$ with respect to the totalized grading.

\begin{rmk}
	While we will not need any technology from the theory of derived algebraic geometry in this paper, we refer the interested reader to the article \cite{ToenDAG} of To\"en for a survey.
\end{rmk}

The functor in \eqref{informal_equivalence} from $\mathrm{Mod}_{C^\bullet(\fn)}^{\fg_0}$ to multiplets is defined by extending the construction in~\cite{perspectives}. It is closely connected to more standard versions of Koszul duality in at least two ways.
First of all, it can be thought of as arising from the dual pair $(C^\bullet(\fn), U(\fn))$. The standard Koszul duality functor is generated by the kernel
\begin{equation*}
(\mathcal K', \mathcal D') = \big(C^\bullet(\fn) \otimes U(\fn) , X^i \otimes n_i\big),
\end{equation*}
where $\big\{X^i\big\}$ is a basis for $\fn^*$, therefore a set of generators for $C^\bullet(\fn)$, and $n_i$ denotes the corresponding
basis of~$\fn$ (thus set of generators for $U(\fn)$).

If we now view $C^\infty(N)$ as a right module for $U(\fn)$, where $\fn$ acts by right-invariant vector fields---in particular, $n_i$ by the vector field $D_i := \rho(n_i)$---we can modify this kernel to obtain a~$(C^\bullet(\fn), C^\infty(N))$-bimodule of the form
\begin{equation*}
(\mathcal K, \mathcal D) = \big(C^\bullet(\fn) \otimes C^\infty(N) , X^i \otimes D_i\big).
\end{equation*}
Our functor is defined as the integral transform associated to this kernel. From this perspective, we can interpret the functor heuristically in two steps:
\begin{enumerate}\itemsep=0pt
	\item First form the Koszul dual $U(\fn)$-module of a module over the Chevalley--Eilenberg complex.
	\item Then apply the associated bundle construction to obtain a~dg-vector bundle over the supergroup $N$, with a residual right $N$-action. In particular, we can forget down to a~dg-vector bundle on the even part $V$ of $N$, with a geometric action of the supersymmetry algebra $\fn$.
\end{enumerate}

It is another version of Koszul duality which is perhaps most closely connected. Kapranov~\cite{Kapranov} established a version of Koszul duality that provides a Quillen equivalence between the model categories of $D$-modules on a space and $\Omega^\bu$-modules over the same space. Our construction can be understood as a version of this equivalence in the case where that space is the nilpotent super Lie group $N$; it relates translation-invariant natural vector bundles to modules over the translation-invariant differential forms, which are precisely the Lie algebra cochains. The connection to the notion of ``multiplet'' in the physics literature has, as far as we know, not been appreciated before.

\begin{rmk}
	While Koszul duality provides a natural language in which to understand the pure spinor construction, our proof of Theorem~\ref{thm: equiv}, and all of our discussion of examples, will proceed by direct computation using the kernel $(\cK,\cD)$. This is meant to emphasize that the technique provides not just an abstract equivalence, but a set of efficient and practical computational techniques.
	We also emphasize that the version of Koszul duality encapsulated in the pure spinor formalism naturally produces dg vector bundles whose sections are sheaves on the site of manifolds equipped with appropriate structure: in the case of standard super-Poincar\'e algebras, this means that one knows how to place the resulting multiplet on any Riemannian manifold. We see this as being connected to ideas in Cartan geometry for the model space $N \cong G/G_0$, and hope to return to versions of the formalism on non-flat spacetimes (for nontrivial or more general Cartan geometries) in future work.
\end{rmk}

Finally, it is worth remarking that the functor from right to left in~\eqref{informal_equivalence} is also easy to understand: it is just taking the derived $\fn$-invariants of a multiplet. This is intuitively satisfying in various respects. For example, the twist of a multiplet consists of the derived invariants of some abelian odd subalgebra of~$\fn$, whereas its dimensional reduction consists of the invariants of an abelian even subalgebra of~$\fn$. The fact that a multiplet can be recovered from the datum of its derived $\fn$-invariants thus says, in a sense, that it is equivalent to all of its possible twists, considered as a natural family over the derived classifying space $B\fn$.

The upshot of this result is that, after an appropriate derived upgrade, \emph{every} supermultiplet has a pure spinor superfield description.
We can use the equivalence of categories to shed some light on the motivating questions concerning the underived pure spinor formalism discussed above. For example, the following is an immediate consequence:

\begin{thm*}[Corollary \ref{underived_recognition_cor}]
	A given multiplet $(E,D,\rho)$ lies in the image of the \emph{underived} pure spinor formalism if and only if the Chevalley--Eilenberg cohomology $\mathrm H^\bullet(\fn , \cE)$ is concentrated in a single degree.
\end{thm*}

We illustrate the derived formalism with some applications and examples; we also show that the derived formalism gives rise to a construction of certain maps relating the multiplets associated to the different Chevalley--Eilenberg cohomology groups of the supertranslation algebra, associated to a filtration on $C^\bu(\fn)$. We give explicit calculations for examples with minimal supersymmetry in dimensions three and four.

\subsection*{Further directions}
The realization of the pure spinor formalism as an equivalence of dg-categories offers potential insight to numerous more involved applications. We give a few indications of possible directions here:
\begin{enumerate}\itemsep=0pt
	\item Given any multiplet $M=(E,D,\rho)$, associated to any supersymmetry algebra $\fg$ with associated supertranslation algebra $\fn$, it is possible to realize a $C^\bullet(\fn)$-module $\Gamma$ generating~$M$ via the pure spinor formalism: one can simply set $\Gamma = C^\bullet(\fn, \cE)$. These modules have cohomology consisting of a graded sum of finitely generated modules over the ring of functions on the nilpotence variety, together with additional information encoding the action of the generators of~$C^\bu(\fn)$ in nonzero total degree. By the results of~\cite{spinortwist}, they efficiently and fully encode the information of the various \emph{twists} of the original multiplet~$M$: one can compute the stalk of the module at a classical point $Q \colon C^\bullet(\fn, \cE) \to \CC$, obtaining descriptions of different twists for different orbits under the action of $\Spin(n)$ and the group of~$R$-symmetries. We refer to the discussion in \cite{NV, ElliottSafronov, ESW} for details of this classification.
	
	It would, for example, be interesting to explicitly understand the modules associated to multiplets such as the $\mathcal N = (1,0)$ tensor multiplet in six dimensions.
	
	\item One can use the Batalin--Vilkovisky (BV) formalism to construct \emph{interacting} supersymmetric classical field theories using the pure spinor formalism. In the BV formalism, an~interacting supersymmetric field theory is encoded by a cyclic $L_\infty$ structure on a multiplet~$M$, that is, an $L_\infty$ algebra structure together with a pairing of degree $-3$. This additional data can be carried along the pure spinor functor; it is enough to define a Lie algebra structure and shifted symplectic pairing internal to equivariant $C^\bullet(\fn)$-modules. This will require a bit more input; it is not straightforward to define a natural tensor product on the category of multiplets, so one would need to work instead with the $D$-modules obtained by taking the sheaf of sections of the multiplet with its natural tensor product. One would then aim use the equivalence to establish a ``convolution'' monoidal structure $\ast$ making the functor monoidal.
	
	For example, as computed in~\cite{CederwallM5} and discussed in~\cite[Section~7.3]{spinortwist} and~\cite{MaxTwist}, if $\fn$ is the eleven-dimensional $\mathcal N=1$ supertranslation algebra, the eleven-dimensional supergravity multiplet arises from the zeroth cohomology of $C^\bullet(\fn)$, placed in degree $-3$ with respect to the natural weight grading, as a $C^\bullet(\fn)$-module. In order to obtain an \emph{interaction}, it would suffice to define, up to homotopy, a Lie structure on this module. It would be interesting to understand the action functionals of \cite{Ced-11d, Ced-towards} in this language, and to use them to connect to component actions for perturbative supergravity, either twisted or untwisted. An interacting BV theory conjecturally describing the minimal twist of eleven-dimensional supergravity was studied in~\cite{RSW}.

	\item There are several sheaves over the (derived) nilpotence variety that automatically carry Lie structures. For example, if $\fh$ is a finite-dimensional semisimple Lie algebra, one can form the tensor product $\mathrm H^0(\fn) \otimes \fh$.\footnote{When we discuss Lie algebra cohomology $\mathrm H^\bu(\fn)$, we will view $\fn$ as being a $\ZZ$-graded object with respect to the weights of the natural rescaling action. As such, the complex $C^\bullet(\fn)$ is naturally bigraded; we refer here to the zeroth cohomology with respect to the total grading. We will discuss these degree conditions further in~Section~\ref{CE_module_section}.} It has been known for some time that this Lie algebra describes interacting ten-dimensional super Yang--Mills theory in the BV formalism; the correct gauge algebras for lower-dimensional super Yang--Mills theories are also obtained in this fashion.
	
	In dimensions higher than eleven we expect $C^\bullet(\fn)$ to have no higher cohomology, so $\Spec (C^\bullet(\fn))$ will be purely classical. Equivalently, the nilpotence variety is expected to be a complete intersection according to Hartshorne's conjecture: roughly speaking, the nilpotence variety is determined by a system of $n$ equations, while the number of variables is of order $2^{n/2}$. Hartshorne's conjecture states that any smooth projective variety in $\mathbb P^n$ with codimension $<n/3$ is a complete intersection. The codimension condition applies to nilpotence varieties with minimal supersymmetry in dimension $\ge 12$, although these varieties are not smooth (but see for instance the recent article \cite{ErmanSamSnowden} for version of this result that would be applicable to the example of nilpotence varieties in high dimensions). However, it is often possible to obtain non-trivial cohomology by taking---rather than the entire Chevalley--Eilenberg complex---a non-trivial quotient associated to an orbit closure within the nilpotence variety.
	
	To give a further example, the tangent sheaf $\mathrm{Der}(C^\bullet(\fn))$ carries a Lie bracket for all supertranslation algebras $\fn$. While the associated multiplet is not generally associated to a BV theory---it does not typically carry a $-1$-shifted symplectic pairing---one can always build a BV theory by applying the cotangent theory construction and considering the multiplet $M \otimes M^![-1]$, where $M^! = M^* \otimes \Omega^{\mathrm{top}}(V)$ is the density-valued dual. This procedure allows for the construction of a very general family of interacting classical supersymmetric field theories.
	
	\item The equivalence of categories allows for a program to classify families of multiplets starting from the algebraic geometry of the (derived) nilpotence variety. This direction is explored in~\cite{HNSW6d}, where (among other things) a description of all six-dimensional $\cN=(1,0)$ multiplets whose derived invariants form a single line bundle on the projective nilpotence variety is given.
\end{enumerate}

\section{The category of multiplets} \label{sec: multiplets}
In~\cite{perspectives}, the notion of a \emph{$\fg$-multiplet} for a super Lie algebra $\fg$ was formalized. Let us quickly review the relevant definitions, and then move on to define morphisms between $\fg$-multiplets and study the resulting category $\Mult_{\fg}$. We start with the case of strict multiplets and strict morphisms and then discuss the homotopy-theoretic generalization.

In brief, by a \emph{multiplet} on a smooth manifold $X$ we will mean a graded vector bundle $E$ on $X$ together with a differential $D$ and an action $\rho$ of an appropriate super Lie algebra $\fg$, compatible with an action of the even part of $\fg$ on the base manifold $X$. The reader should consider the prototypical situation where $X = \RR^n$, and $\fg$ is a super Lie algebra whose even part is the Poincar\'e algebra of infinitesimal isometries of $\RR^n$. A multiplet can be thought of as modelling the perturbations of a classical field in the BV formalism.

\begin{rmk}
	We would like to clearly emphasize here the distinction between multiplets and theories in this context. In general, the data of a multiplet is \emph{less data} than the data of a classical field theory: it does not include enough data to specify dynamics. The data of a perturbative classical supersymmetric field theory in the BV formalism includes the data of a multiplet, together with some additional data. If a multiplet is equipped with the additional data of a~\emph{BV bracket} (an odd shifted symplectic structure), it can be viewed as modelling a \emph{free} theory in the BV formalism. Such structures can easily be constructed from multiplets by taking the sum of a multiplet with a shift of its dual.	If one wishes to include non-trivial interactions, one must additionally equip the sections of the bundle $E$ with an $L_\infty$ structure, compatibly with the supersymmetry action.
	
	So, while the results of this paper do apply to BV field theories, when we speak of \emph{multiplets} we are referring to a less constrained notion, and to obtain a genuine field theory from a multiplet one needs to specify strictly more data.
\end{rmk}

\subsection{Gradings}
\label{ssec:gradings}
Many objects appearing throughout this work (e.g., vector spaces, vector bundles, Lie algebras etc.) will carry a grading by $\ZZ\times \ZZ/2\ZZ$ as well as a differential of bidegree $(1,+)$. We refer to such objects with the prefix ``dgs'' standing for ``differential graded super''. Typically these gradings have a clear physical meaning: the integer grading refers to the ghost number or cohomological degree, while the $\ZZ/2\ZZ$-grading corresponds to the intrinsic parity (fermion number modulo two). The total parity denotes the $\ZZ/2\ZZ$-grading which arises by forgetting the $\ZZ$-grading to~$\ZZ/2\ZZ$ and then totalizing with the intrinsic parity. It is the total parity which governs all signs.

It will often be convenient to introduce an additional piece of data: an auxiliary action of a one-dimensional abelian Lie algebra $\RR$ on a dgs vector space. Many of our examples will be built from the following simple observation.

\begin{eg}
	Let $\fn$ be a dg Lie algebra concentrated in degrees 1 and 2. For all non-negative integers $k$, there is a natural action of $\RR$ on $\Sym^k(\fn^*)$, where $\RR$ acts on duals to the degree 1 and 2 summands of $\fn^*$ with weight $-1$ and $-2$ respectively.
\end{eg}

The results in this paper will make sense for arbitrary $\RR$-equivariant dgs vector spaces, but our most common examples will have the following behavior.

\begin{dfn}
	A \emph{lifted} dgs vector space is an $\RR$-equivariant dgs vector space where the $\RR$-action has integral weights, and the $\ZZ/2\ZZ$-grading coincides with the $\RR$-weight modulo two.
\end{dfn}

Let us now introduce some notation that we will use when we discuss lifted dgs vector bundles on a smooth manifold $X$. Let us write
\begin{equation*}
E = \bigoplus_{(w,d) \in \ZZ^2} E^{w,d}
\end{equation*}
for a lifted dgs vector bundle, where the first index indicates the decomposition of $E$ by $\RR$-weight, and the second index indicates internal $\ZZ$-degree on $E$. In some of our calculations, particular in~Section~\ref{sec: examples}, we will place the summands in a two-dimensional array, where the two coordinates are given by $w-d$ and $d$. For example,
\begin{equation*}
E =
\begin{bmatrix}
\cdots &E^{-1,0} &E^{0,0} &E^{1,0} &E^{2,0}&\cdots\\
\cdots &E^{0,1} &E^{1,1} &E^{2,1} &E^{3,1}&\cdots\\
\cdots &E^{1,2} &E^{2,2} &E^{3,2} &E^{4,2} &\cdots
\end{bmatrix}.
\end{equation*}
So we recover $w$ as the total degree with respect to this sheared bigrading. The overall parity is then determined by the column.

If the differential on our dgs vector bundle decomposes as a sum of homogeneous differential operators of order $k$, say $D = \sum_{k \ge 0} D_k$, the summand $D_k$ acts---with respect to the sheared grading---with bidegree $(2k-1, 1)$. So the homogeneous summands all increase the vertical degree by 1, but they may modify the horizontal degree by any odd integer $\ge -1$.

\subsection{Strict multiplets}

For a dgs vector bundle\footnote{Note that each graded piece $E^k$ is a finite rank vector bundle, but the total rank of $E$ may still be infinite.} $(E,D)$ over a base smooth manifold $X$, we denote the space of global smooth sections by $\cE = \Gamma(X,E)$. The endomorphisms $\End(\cE)$ form a dgs Lie algebra, where the bracket is given by the commutator and the differential is $[D,-]$. Inside $(\End(\cE), [D,-])$ there is a sub dgs Lie algebra consisting of all endomorphisms which act on sections via differential operators. We denote this subalgebra by $(\cD(E), [D,-])$.
\begin{dfn}
	A \emph{strict local dgs $\fg$-module} is a triple $(E,D,\rho)$ where $(E,D)$ is a dgs vector bundle and
	\begin{equation*}
	\rho \colon \ \fg \longrightarrow (\cD(E) , [D,-] )
	\end{equation*}
	is a map of dgs Lie algebras. Here the super Lie algebra $\fg$ is viewed as a dgs Lie algebra in cohomological degree zero with trivial differential.
\end{dfn}
\begin{rmk}
	This definition (as well as many of the following) has a natural generalization for~$\fg$ a dgs Lie algebra. Since we are ultimately interested in the pure spinor superfield formalism, we restrict our attention to super Lie algebras with no cohomological grading.
\end{rmk}
Note that, since a super Lie algebra $\fg$ has no differential, $\rho$ commutes with the differential on the dgs vector bundle,
\begin{equation*}
[D,\rho(x)] = 0 \qquad \forall x \in \fg .
\end{equation*}
It is standard to encode a $\fg$-module structure $\rho$ on $(E,D)$ as a dgs Lie algebra structure on the direct sum $\fg \oplus \cE$. Concretely we set for the unary and binary operations
\begin{align}
&\mu_1(x_1,\sigma_1) = (0, D\sigma_1),\nonumber \\
&\mu_2((x_1,\sigma_1),(x_2,\sigma_2)) = \big([x_1,x_2] , \rho(x_1)\sigma_2 - (-1)^{|\sigma_1| |x_2|} \rho(x_2)\sigma_1 \big) ,\label{eq: lie structure on sum}
\end{align}
where $x_1,x_2 \in \fg$ and $\sigma_1,\sigma_2 \in \cE$.

There is an obvious notion of morphisms between strict local $\fg$-modules.
\begin{dfn}
	A \emph{strict morphism} of strict local $\fg$-modules	$(E,D,\rho)$ and $(E',D',\rho')$ is a map of cochain complexes
	\begin{equation*}
	\psi \colon \ \cE \longrightarrow \cE'
	\end{equation*}
	realized by differential operators such that
	\begin{equation*}
	\psi \circ \rho(x) = \rho'(x) \circ \psi
	\end{equation*}
	for all $x \in \fg$.
\end{dfn}
A strict morphism $\psi \colon (E,D,\rho) \longrightarrow (E',D',\rho')$ gives rise to a strict morphism of the associated dgs Lie algebras by setting $\tilde{\psi} = \id_{\fg} \times \psi \colon \fg \oplus \cE \longrightarrow \fg \oplus \cE'$. Conversely it is easy to check that every strict morphism of dgs Lie algebras of that form gives rise to a strict morphism of $\fg$-modules. We call $\psi$ a quasi-isomorphism if it is a quasi-isomorphism of dgs vector bundles; equivalently $\tilde{\psi}$ is a quasi-isomorphism of dgs Lie algebras.

Now, suppose that $(E,D)$ is a dgs vector bundle over an affine space $V$ (so, in the notation used above, $V=X$). Let us additionally assume that $\fg$ is equipped with a map of super Lie algebras
\begin{equation*}
\phi \colon \ \aff(V) \longrightarrow \fg .
\end{equation*}
In essence, when we refer to strict $\fg$-\emph{multiplets}, we are referring to strict local $\fg$-modules for which the affine transformations acts in a geometric way.
\begin{dfn}
	A \emph{strict $\fg$-multiplet} $(E,D,\rho)$ on $V$ is an affine dgs vector bundle over $V$ equipped with a strict local $\fg$-module structure
	\begin{equation*}
	\rho \colon \ \fg \longrightarrow \cD(E)
	\end{equation*}
	such that the pullback of the module structure along~$\phi$ agrees with the natural action on sections of the affine vector bundle. Concretely, this means that the following diagram commutes:
	\begin{equation} \label{eq: geom action}
	\begin{tikzcd}
	\fg \arrow[r, "\rho"] & \cD(E). \\
	\aff(V) \arrow[ru, "\mathrm{aff}" '] \arrow[u,"\phi"]
	\end{tikzcd}
	\end{equation}
\end{dfn}
Here $\mathrm{aff}$ denotes the natural action of the affine algebra on $E$.
We define the category of strict multiplets with strict morphisms to be the full subcategory of the category of strict local $\fg$-modules with objects strict $\fg$-multiplets. We denote this category by $\Mult_{\fg}^{\mathrm{strict}}$.
\begin{rmk}
	The condition~\eqref{eq: geom action} stating that the translations act geometrically heavily depends on $\phi$. In particular, since the action $\mathrm{aff}$ is effective, the corresponding category of $\fg$-multiplets is trivial if $\phi$ is not injective.
\end{rmk}
\begin{rmk}
	Note that a morphism of strict local $\fg$-modules is automatically compatible with the action of $\aff(V)$. For example let $\psi \colon (E,D,\rho) \longrightarrow (E',D',\rho')$ be a strict morphism. Then we have
	\begin{equation*}
	\psi \circ \mathrm{aff}(x) = \psi \circ \rho(\phi(x)) = \rho'(\phi'(x)) \circ \psi = \mathrm{aff}'(x) \circ \psi .
	\end{equation*}
	Therefore it is sensible to define $\Mult_\fg^{\text{strict}}$ as a full subcategory of strict local $\fg$-modules.
\end{rmk}
\subsection{Homotopy theory of multiplets}
It is well known that various strict algebraic structures, such as associative algebras or Lie algebras, admit homotopy-theoretic generalizations ($A_\infty$ or $L_\infty$ algebras in these cases). The same is true for strict $\fg$-module structures,
which generalize to $L_\infty$ $\fg$-modules. In the context of this work, we will refer to such homotopy module structures simply as module structures and choose to emphasize whenever a module is strict instead. In this spirit we can easily define (not necessarily strict) local $\fg$-modules and $\fg$-multiplets by replacing Lie maps in the above definitions with $L_\infty$ maps.

Thus, a local $\fg$-module is just a dgs vector bundle $(E,D)$ together with an $L_\infty$ map of $L_\infty$ algebras
\begin{equation*}
\rho \colon \ \fg \rightsquigarrow \cD(E) .
\end{equation*}
Recall that this means there are component maps
\begin{equation*}
\rho^{(k)} \colon \ \fg^{\otimes k} \longrightarrow \cD(E)[1-k], \qquad k \geq 1
\end{equation*}
satisfying a series of compatibility relations the lowest of which reads
\begin{equation*}
\big[\rho^{(1)}(x),\rho^{(1)}(y)\big] - \rho^{(1)}([x,y]) =\big[D, \rho^{(2)}(x,y)\big] \qquad \forall x,y \in \fg.
\end{equation*}
Clearly a local $\fg$-module is strict if and only if $\rho^{(k)} = 0$ for all $k\geq2$.

\begin{rmk}
	As we discussed in the introduction, considering this homotopical version of a~module for a super Lie algebra is very natural in supersymmetric field theory. In most supersymmetric classical field theories, when one describes the most natural multiplet of fields, there is only a strict action of the supersymmetry algebra on-shell, i.e., after imposing the equations of motion. If one, however, allows $L_\infty$ actions, it is possible to describe an $L_\infty$ action of the supersymmetry algebra on the space of fields in the BV formalism directly, without introducing auxiliary fields. This idea goes back to the beginnings of the BV formalism; applications to global symmetries are discussed, for example, in~\cite{BaulieuSusy}. The example of super Yang--Mills theory is discussed rigorously and systematically in all dimensions in~\cite{ESW}.
\end{rmk}

Similarly to the strict case, we can conveniently encode a $\fg$-module structure $\rho$ as an~$L_\infty$ structure on $\fg \oplus \cE$. To this end we supplement the operations~\eqref{eq: lie structure on sum} by the following brackets for $k\geq3$ (for details we refer to~\cite{Allocca, Lada})
\begin{equation*}
\mu_k((x_1,v_1),\dots,(x_k,v_k)) = \left( 0 , \sum_{i=1}^{k} \pm \rho^{(k-1)}\big(x_1,\dots,\hat{x}_i,\dots,x_k\big)v_i \right).
\end{equation*}
One can define a morphism $\psi \colon (E,D,\rho) \longrightarrow (E',D',\rho')$ of local $\fg$-modules by component maps
\begin{equation*}
\psi_n \colon \ \fg^{\otimes n-1} \otimes \cE \longrightarrow \cE' ,
\end{equation*}
which are given by differential operators and satisfy a series of compatibility relations (see~\cite{Allocca} for details). We recover strict morphisms by restricting to those where the only component map is $\psi_1$. Again, we can describe such morphisms by morphisms of the associated $L_\infty$ algebras
\begin{equation*}
\tilde{\psi} \colon \ \fg \oplus \cE \rightsquigarrow \fg \oplus \cE'
\end{equation*}
by supplementing $\tilde{\psi}_1 = \id_{\fg} \times \psi_1$ with
\begin{equation*}
\tilde{\psi}_k ((x_1,\sigma_1) , \dots , (x_k, \sigma_k)) = \left(0, \sum_{i=1}^{k} \pm \psi_k\big(x_1, \dots , \hat{x}_i , \dots , x_k , \sigma_i \big)\right)
\end{equation*}
for $k \geq 2$. $\psi$ is a quasi-isomorphism of local $\fg$-modules if $\psi_1$ is a quasi-isomorphism of cochain complexes, or equivalently if $\tilde{\psi}$ is a quasi-isomorphism of $L_\infty$ algebras. In general, encoding module structures as $L_\infty$ structures is very convenient because it allows the use of many known tools from the theory of $L_\infty$ algebras, like homotopy transfer for instance.

\begin{rmk}
	An equivalent realization of the notion of a non-strict morphism between a pair of local $L_\infty$ algebras $\fg,\fg'$ is provided by considering the Chevalley--Eilenberg chain complex~$C_\bullet(\fg)$ of an $L_\infty$ algebra $\fg$. This is a cocommutative dg coalgebra whose underlying graded coalgebra is $\Sym^\bullet(\fg[-1])$, with differential induced from the $L_\infty$ structure on $\fg$ as a sum of terms given by the $L_\infty$ brackets on $\fg$. The data of a morphism $\psi \colon \fg \to \fg'$ of local $L_\infty$ algebras is equivalent to that of a morphism of cocommutative dg coalgebras
	\[\psi' \colon \ C_\bullet(\fg) \to C_\bullet(\fg')\]
	given by differential operators. This interpretation allows us to view local $L_\infty$ algebras as objects of a dg-category, so that we can discuss (for example) homotopies between $L_\infty$ algebra morphisms. When the source $\fg$ is finite-dimensional, we can dually consider morphisms of commutative dg algebras $C^\bullet(\fg') \to C^\bullet(\fg)$. This is for example the case for the local $\fg$-module structures $\rho \colon \fg \rightsquigarrow \cD(E)$ appearing in the definition of multiplets.
\end{rmk}

Let us now give the definition of a (not necessarily strict) multiplet.
\begin{dfn}
	A \emph{$\fg$-multiplet} $(E,D,\rho)$ is an affine dgs vector bundle over $V$ equipped with a local $\fg$-module structure
	\begin{equation*}
	\rho\colon \ \fg \rightsquigarrow \cD(E)
	\end{equation*}
	such that the pullback of the module structure along~$\phi$ agrees strictly with the natural action on sections of the affine vector bundle. Concretely, this means that the following diagram commutes:
	\begin{equation*}
	\begin{tikzcd}
	\fg \arrow[r, "\rho^{(1)}"] & \cD(E). \\
	\aff(V) \arrow[ru, "\mathrm{aff}" '] \arrow[u,"\phi"]
	\end{tikzcd}
	\end{equation*}
	We define the dg-category of $\fg$-multiplets to be the full subcategory of local $\fg$-modules with objects being $\fg$-multiplets and denote it by $\Mult_{\fg}$.
\end{dfn}

\begin{rmk}
	Note that, in particular, the action $\rho \circ \phi$ obtained by restricting the $L_\infty$ action to affine transformations is necessarily strict.
\end{rmk}

We will sometimes wish to refer to the full subcategory of strict multiplets, but allowing arbitrary morphisms.
\begin{dfn}
	Denote by $\Mult_{\fg}^{\text{\rm strict-ob}}$ the full sub dg-category of $\Mult_{\fg}$ generated by strict multiplets.
\end{dfn}

We can obtain homotopy categories from the dg-categories $\Mult_{\fg}$ as well as $\Mult_{\fg}^{\mathrm{strict}}$ by replacing the hom space $\Hom_{\Mult_{\fg}}(E, E')$ by its zeroth cohomology $\mathrm H^0(\mathrm{Hom}_{\Mult_{\fg}}(E, E'))$. We denoted the resulting categories by $\mathrm{Ho}(\Mult_{\fg})$ and $\mathrm{Ho}(\Mult_{\fg}^{\mathrm{strict}})$. Speaking physically, quasi-isomorphisms of $\fg$-multiplets correspond to \emph{perturbative} equivalences of multiplets, where by ``perturbative'' here we mean equivalences of the derived formal neighborhoods of a point in the classical moduli field space. Therefore isomorphism classes in the homotopy category correspond to perturbatively distinct multiplets.

\begin{rmk}
	Because every $L_\infty$ algebra can be strictified \cite[Part II, Corollary 1.6]{KrizMay}, the natural inclusion
	\begin{equation*}
	\mathrm{Ho}\big(\Mult_{\fg}^{\text{\rm strict-ob}}\big) \to \mathrm{Ho}(\Mult_{\fg})
	\end{equation*}
	is an equivalence of categories, where $\Mult_{\fg}^{\text{\rm strict-ob}}$ denotes the full subcategory of multiplets spanned by strict objects.
\end{rmk}

\begin{eg}
	Let us give one example of a non-strict multiplet for three dimensional $\cN=1$ supersymmetry. Recall that $\Spin(3) \cong \mathrm{SU}(2)$; we denote the two dimensional spinor representation by $S$ and the three-dimensional vector representation by $V$. We fix $\fg$ to be the super Poincar\'e algebra, whose underlying $\ZZ/2\ZZ$-graded vector space is of the form
	\begin{equation*}
	\fg = (\so(3) \oplus V) \oplus \Pi S
	\end{equation*}
	where $\Pi$ indicates the shift into odd parity. The symmetric bracket is induced from the isomorphism $\Gamma \colon\Sym^2(S) \cong V$ of $\mathrm{SU}(2)$-representations. Let us define $E$ to be the trivial vector bundle over $V = \RR^3$ with fibers
	\begin{equation*}
	E^{-1}_+ = \CC, \qquad E^0_+ = V, \qquad E^0_- = S ,
	\end{equation*}
	where we have used the subscript $\pm$ to indicate $\ZZ/2\ZZ$-degree. The field content consists of a~zero-form, a one-form and a fermion field. The differential $D$ operates on sections as the de Rham differential $\d\colon \Omega^0 \longrightarrow \Omega^1$, and vanishes elsewhere. This dgs vector bundle can be lifted to a $\ZZ \times \ZZ$-graded vector bundle where $E^0_+$ has weight one, and $E^0_-$ has weight two.
	
	We summarize this field content with the below array, using the conventions of~Section~\ref{ssec:gradings}:
	\begin{equation*}
	\left[
	\begin{tikzcd}[column sep = 0.3cm , row sep = 0.3cm]
	\Omega^0 \arrow[dr ,"\d"] & & \\
	& \Omega^1 & S
	\end{tikzcd} \right].
	\end{equation*}
	The even part of $\fg$ acts in the standard geometric fashion. For $Q\in \fg_-$ we set
	\begin{alignat*}{4}
	&\rho^{(1)}(Q) \colon \quad&& S \longrightarrow \Omega^1 ,\qquad&& \psi \mapsto \Gamma(Q,\psi),& \\
	&&& \Omega^1 \longrightarrow S ,\qquad&& A \mapsto Q \wedge \slashed{\partial} A,& \\
	&\rho^{(2)}(Q,Q) \colon\quad&& \Omega^1 \longrightarrow \Omega^0 ,\qquad&& A \mapsto \iota_{[Q, Q ]} A.&
	\end{alignat*}
	Here $\iota$ denotes the contraction of a differential form by a vector field and we view $[Q,Q]$ as a~constant vector field on $X$.
\end{eg}

\subsection{Linear structure}

We can define the direct sum of two $\fg$-multiplets as follows.
\begin{dfn}
	The \emph{direct sum} of two $\fg$-multiplets $(E,D,\rho)$ and $(E',D',\rho')$ is defined to be the multiplet
	\begin{equation*}
	(E,D,\rho) \oplus (E',D',\rho') = (E\oplus E', D \oplus D', \rho \oplus \rho') .
	\end{equation*}
\end{dfn}

\begin{rmk}
	In contrast, defining a tensor product on the category of $\fg$-multiplets is not straightforward. Considering the D-modules of global sections, one can take the tensor product in the category of D-modules, but this does not take the additional structures on a multiplet into account. As alluded to in the introduction, one can might also try using the equivalence established in Section~\ref{sec: equiv} to define a product on the category which makes the functor monoidal. We will not discuss this issue further in this work.
\end{rmk}



For further reference, we also define the dual of a $\fg$-multiplet. In physics context, these are usually called antifield multiplets
\begin{dfn} \label{def: antifield}
	The \emph{dual} of a $\fg$-multiplet $(E,D,\rho)$, also referred to as the associated \emph{antifield multiplet} is the $\fg$-multiplet $\big(E^!,D^*, \rho^*\big)$. Here $E^!$ is the linear dual vector bundle to $E$ twisted by the canonical bundle. The action $\rho^*$ denotes the map
	\begin{equation*}
	\rho^* \colon \ \fg \rightsquigarrow \cD\big(E^!\big)
	\end{equation*}
	given by $\rho^{* (k)}(Q_1,\dots,Q_k) = \rho^{(k)}(Q_1,\dots, Q_k)^*.$
\end{dfn}

Note that we will typically be working over the flat space $\RR^n$, so we may choose a trivialization of the canonical bundle if we wish to identify $E^!$ with $E^*$.

\section{Derived pure spinor superfields}
In this section, we define the derived generalization of the pure spinor superfield construction, which will provide one of the two functors that witness the equivalence of categories. We begin by setting up the general context in which we want to work.

As in the introduction, let $\fn$ be a two-step nilpotent super Lie algebra, defined by a central extension
\begin{equation}\label{eq:2step}
	0 \to \fn_2 \to \fn \to \Pi \fn_1 \to 0
\end{equation}
of the odd abelian super Lie algebra $\Pi \fn_1$. We imagine such objects as generalizations of supertranslation algebras. The automorphisms of $\fn$, which are all outer, will contain an abelian factor generating scale transformations, with respect to which $\fn_1$ has weight one and $\fn_2$ weight two. There is also a natural map
\begin{equation*}
	\aut(\fn) \to \lie{gl}(\fn_2).
\end{equation*}
The kernel of this map can be thought of as the $R$-symmetry algebra; in physical examples, $\aut(\fn)$ will be the product of $\so(\fn)$, scale transformations, and the $R$-symmetry algebra.

All of our constructions will take place in reference to a fixed super Lie algebra $\fg$ of the following type:
\begin{dfn}
	A super Lie algebra $\fg$ is of \emph{super Poincar\'e type} if it can be written as an extension
	\begin{equation*}
	0 \to \fn \to \fg \to \fg_0 \to 0,
	\end{equation*}
	where $\fn$ is a two-step nilpotent super Lie algebra of the form~\eqref{eq:2step} and $\fg_0$ is a Lie algebra equipped with a Lie map $\fg_0 \to \aut(\fn)$.
\end{dfn}

This means that the $\ZZ/2\ZZ$ grading on $\fg$ can be lifted to a $\ZZ$-grading concentrated in degrees zero, one, and two:
\begin{equation*}
\fg = \fg_0 \oplus \fn_1 \oplus \fn_2.
\end{equation*}
The most important examples are the super Poincar\'e algebras in various dimensions and with various amounts of supersymmetry.

It is possible to consider examples where $\fg_0$ is any subalgebra of~$\aut(\fn)$. We will typically wish, however, to take $\fg_0$ to
be such that any linear transformation of $\fn_2$ arising from an automorphism of~$\fn$ can be generated by an element of~$\fg_0$.
In the super-Poincar\'e case, this will mean that $\so(n)$ is contained in~$\fg_0$. In this paper, we will always simply take $\fg_0 = \aut(\fn)$.

Let us choose a basis $d_\alpha$ for $\fn_1$ and $e_\mu$ for $\fn_2$ such that we can expand the symmetric bracket in structure constants\footnote{In the context of super Poincar\'e algebras, these structure constants are typically expressed in terms of the matrix elements of the gamma matrices.}
\begin{equation*}
\big[d_\alpha , d_\beta \big] = f^\mu_{\alpha \beta} e_\mu .
\end{equation*}
We denote by $R = \Sym^\bullet\big(\fn_1^*\big) = \CC[\lambda^\alpha]$ the ring of polynomial functions on $\fn_1$. For $Q$ in $\fn_1$, the equation $[Q,Q] = 0$ defines an ideal $I$ in $R$, explicitly given by
\begin{equation*}
I = \big(\lambda^\alpha f^\mu_{\alpha \beta} \lambda^\beta \big) .
\end{equation*}
As we will discuss shortly, we can identify the quotient algebra $R/I$ with the degree zero Lie algebra cohomology $\mathrm H^0(\fn)$. Here, by degree zero, we mean with respect to the totalization of the canonical $\ZZ \times \ZZ$-grading on Chevalley--Eilenberg cochains.

The pure spinor superfield formalism gives a systematic tool to construct $\fg$-multiplets from the input datum of a graded $\fg_0$-equivariant $R/I$-module. Here we generalize the pure spinor superfield construction to $\fg_0$-equivariant $C^\bullet(\fn)$-modules and show that it defines a functor.

\subsection[The category of C\^{}bullet(n)-modules]{The category of $\boldsymbol{C^\bullet(\fn)}$-modules} \label{CE_module_section}
Recall that the Chevalley--Eilenberg complex of $\fn$ takes the form
\begin{equation*}
C^\bullet(\fn) = \big( \Sym^\bullet\big(\fn^*[1]\big) , \d_{\mathrm{CE}} \big),
\end{equation*}
where the Chevalley--Eilenberg differential $\d_{\mathrm{CE}}$ is induced by the dual of the bracket. Here the notation $\fn^*[1]$ means that we shift $\fn^*$ down in cohomological degree by one. The Chevalley--Eilenberg complex has a $\ZZ \times \ZZ/2\ZZ$-grading endowing it with the structure of a dgs algebra. In our present case, since the the $\ZZ/2\ZZ$ grading of $\fn$ lifts to a $\ZZ$-grading where $\fn_i$ has weight $i$, $C^\bullet(\fn)$ can be given a $\ZZ \times \ZZ$-grading. Totalizing this grading, the generators of $\fn_1^*$ sit in degree $0$ while the generators of $\fn_2^*$ live in degree $-1$. We denote these generators by $\lambda^\alpha$ and $v^\mu$ respectively. We can thus identify{\samepage
\begin{equation*}
C^{-p}(\fn) = \wedge^p \fn_2^* \otimes R
\end{equation*}
with respect to the totalized grading.}

The Chevalley--Eilenberg differential acts on these generators according to
\begin{equation*}
\begin{split}
\d_{\mathrm{CE}} \lambda^\alpha &= 0, \\
\d_{\mathrm{CE}} v^\mu &= \lambda^\alpha f^\mu_{\alpha \beta} \lambda^\beta.
\end{split}
\end{equation*}
In coordinates we will often write the Chevalley--Eilenberg algebra in the form
\begin{equation*}
\left( C^\bullet(\fn), \d_{\mathrm{CE}} \right) = \left( \CC[\lambda^\alpha ,v^\mu] , \d_{\mathrm{CE}} = \lambda^\alpha f^\mu_{\alpha \beta} \lambda^\beta \frac{\partial}{\partial v^\mu} \right) .
\end{equation*}
Using this description, we immediately see that we indeed recover $H^0(\fn) = R/I$.

$C^\bullet(\fn)$ is a dgs algebra, therefore we can consider dgs modules over it.
\begin{dfn}
	A $C^\bullet(\fn)$-module is a dgs vector space $(\Gamma, \d_\Gamma)$ together with a morphism
	\begin{equation*}
	(C^\bullet(\fn) , \d_{\mathrm{CE}}) \longrightarrow (\End(\Gamma) , [\d_\Gamma , -]),
	\end{equation*}
	of dgs algebras.
\end{dfn}

\begin{dfn}
	A morphism of $C^\bullet(\fn)$-modules $(\Gamma , \d_\Gamma)$ and $(\Gamma' , \d_{\Gamma'})$ is a cochain map
	\begin{equation*}
	f \colon \ (\Gamma , \d_\Gamma) \longrightarrow (\Gamma' , \d_{\Gamma'})
	\end{equation*}
	such that
	\begin{equation*}
	f(x \cdot_\Gamma \gamma) = x \cdot_{\Gamma'} f(\gamma) .
	\end{equation*}
	In other words, $f$ is just a morphism of dgs modules.
\end{dfn}
Note that $\fg_0$ acts on both $\fn_1$ and $\fn_2$ and thus also on $C^\bullet(\fn)$.

\begin{dfn}
	A $C^\bullet(\fn)$ module $\Gamma$ is called \emph{$\fg_0$-equivariant} if $\Gamma$ is also a representation of $\fg_0$ and the module structure map is $\fg_0$-equivariant. We further assume that each degree $\Gamma^k$ is finite dimensional as a complex vector space. We will denote the dg-category of $\fg_0$-equivariant $C^\bullet(\fn)$-modules by $\Mod_{C^\bullet(\fn)}^{\fg_0}$.
\end{dfn}

Our main results will take place in this category of $\fg_0$-equivariant modules.

\subsection{The derived pure spinor superfield formalism}

There is a canonical strict multiplet associated to the super Lie algebra $\fg$: the free superfield. Recall that $\fn$ is a two-step nilpotent super Lie algebra and let $N = \exp(\fn)$ be the associated nilpotent super Lie group. Since all such nilpotent super Lie groups are split, $N$ can be chosen to be the total space of a bundle with fiber $\Pi N_1$ over $N_2$, where $N_2$ is a vector space treated as an additive abelian group. (As super vector spaces, $\fn$ and~$N$ are isomorphic; the group operation can be constructed using the Baker--Campbell--Hausdorff formula, which terminates at quadratic order in this case.)

Left and right translations induce two commuting, $\aut(\fn)$-equivariant actions of $\fn$ on $N$ by vector fields
\begin{equation*}
\sL,\, \sR \colon \ \fn \longrightarrow \Vect(N) .
\end{equation*}
Let us choose coordinates $x^\mu$ on the abelian group $N_2$
and $\theta^\alpha$ on $N_1$. In terms of these coordinates, the vector fields given by the action of the odd elements can be described as follows:
\begin{equation*}
\begin{split}
\sR(d_\alpha) &= \frac{\partial}{\partial \theta^\alpha} - f^\mu_{\alpha \beta} \theta^\beta \frac{\partial}{\partial x^\mu}, \\
\sL(d_\alpha) &= \frac{\partial}{\partial \theta^\alpha} + f^\mu_{\alpha \beta} \theta^\beta \frac{\partial}{\partial x^\mu}.
\end{split}
\end{equation*}
In addition, the even elements simply act by derivatives
\begin{equation*}
\sR(e_\mu) = \sL(e_\mu) = \frac{\partial}{\partial x^\mu}.
\end{equation*}
The free superfield is the strict $\fg$-multiplet with $\cE = C^\infty(N)$, vanishing differential, and module structure given by the left translations $\sL$. Note that by ``free'' here, we mean in the sense of the rank one free module over the ring of smooth functions on superspace (``unconstrained''), rather than ``non-interacting''. Recall that all our multiplets are a priori non-interacting: the inclusion of interactions is additional data with which a multiplet can be equipped.

Let us now generalize the pure spinor superfield formalism to a derived setting.
\begin{dfn}
	We define the \emph{pure spinor functor} $A^\bullet \colon \Mod_{C^\bullet(\fn)}^{\fg_0} \longrightarrow \Mult_{\fg}^{\mathrm{strict}}$ by setting
	\begin{equation*}
	A^\bullet(\Gamma) = \big(C^\infty(N) \otimes_\CC \Gamma , \cD\big) ,
	\end{equation*}
	for an object $(\Gamma , \d_\Gamma)$.
	The differential $\cD$ is constructed using the right action $R$ and the $C^\bullet(\fn)$-module structure on $\Gamma$. Explicitly, it takes the form
	\begin{equation*}
	\cD = \lambda^\alpha \sR(d_\alpha) + v^\mu \sR(e_\mu) + \d_\Gamma .
	\end{equation*}
	For a morphism $f\colon\Gamma \longrightarrow \Gamma'$ we define
	\begin{equation*}
	A^\bullet(f) = \id_{C^\infty(N)} \otimes f \colon \ A^\bullet(\Gamma) \longrightarrow A^\bullet(\Gamma') .
	\end{equation*}
\end{dfn}

A few comments are in order.

\begin{enumerate}\itemsep=0pt
	\item The differential $\cD$ squares to zero precisely since $\Gamma$ is a $C^\bullet(\fn)$-module.
	\item We can equip $A^\bullet(\Gamma)$ with the structure of a dgs vector bundle over the spacetime $V := N_2$ by placing $C^\infty(N)$ in cohomological degree zero. Note that in particular the differential is of bidegree $(1,+)$.
	
	\item Further, $\fg$ acts on $C^\infty(N)$ via left translations and on $\Gamma$ by the trivial extension of the $\fg_0$-module structure which was part of the input datum. The tensor product of these two action makes $A^\bullet(\Gamma)$ into a strict multiplet.
	\item It is immediate to check that $A^\bullet(f)$ is a strict morphism of strict multiplets. Since $f$ is a~morphism of $C^\bullet(\fn)$-modules we have
	\begin{equation*}
	A^\bullet(f) \circ \cD = \cD' \circ A^\bullet(f) .
	\end{equation*}
	In addition,
	\begin{equation*}
	A^\bullet(f) \circ \sL = \sL \circ A^\bullet(f)
	\end{equation*}
	obviously follows from the definition.
	\item $A^\bullet$ is additive. The direct sum of two $C^\bullet(\fn)$-modules is mapped to the direct sum of the respective multiplets, $A^\bullet(\Gamma \oplus \Gamma') = A^\bullet(\Gamma) \oplus A^\bullet(\Gamma')$.
\end{enumerate}

This construction is a direct generalization of the pure spinor superfield formalism as described in~\cite{perspectives}. To see this, we notice that the category of graded equivariant of $R/I$-modules sits as a subcategory inside $\Mod_{C^\bullet(\fn)}^{\mathrm{strict}}$, namely precisely as those modules concentrated in cohomological degree zero. Indeed, every $R/I$-module is a $C^\bullet(\fn)$-module by the map
\begin{equation*}
C^\bullet(\fn) = \CC\big[\lambda^\alpha,v^\mu\big] \longrightarrow R/I, \qquad \big(\lambda^\alpha , v^\mu\big) \mapsto \lambda^\alpha .
\end{equation*}
The other way round, let $(\Gamma , \d_\Gamma)$ be a $C^\bullet(\fn)$-module concentrated in cohomological degree zero. Then the differential $\d_\Gamma$ vanishes and $v^\mu$ acts trivially for degree reasons. Therefore one has, for~$\gamma$ any element of $\Gamma$,
\begin{equation*}
0 = \d_\Gamma \big(v^\mu \cdot \gamma\big) = \big(\d_{\mathrm{CE}} v^\mu\big) \gamma = \big(\lambda f^\mu \lambda\big) \gamma ,
\end{equation*}
which endows $\Gamma$ with the structure of an $R/I$-module.

Restricting the functor $A^\bullet$ to graded equivariant $R/I$-modules, we obtain a functor
\begin{equation*}
A^\bullet_{R/I} \colon \ \Mod_{R/I}^{\fg_0} \longrightarrow \Mult_{\fg}^{\mathrm{strict}}
\end{equation*}
where the output simplifies to
\begin{equation*}
A^\bullet_{R/I} (\Gamma) = \left( C^\infty(N) \otimes_\CC \Gamma , \lambda^\alpha \sR(d_\alpha) \right)
\end{equation*}
such that we precisely recover the pure spinor superfield formalism as presented in~\cite{perspectives}.

Further, this is a derived generalization of the pure spinor superfield construction in the sense that $C^\bullet(\fn)$ can be viewed as a derived replacement of the ring $R/I$. We will therefore refer to this construction as the derived pure spinor superfield formalism.

\begin{rmk}
	We can alternatively view this geometrically as a derived enhancement of the affine nilpotence variety $\mathrm{Spec}(R/I)$. We can view a multiplet as arising from a quasi-coherent sheaf over the nilpotence variety, but this point of view requires forgetting some of the data given by the $C^\bullet(\fn)$ action. The philosophy of derived algebraic geometry suggests instead retaining this information by viewing a multiplet as arising from a coherent sheaf over the affine derived scheme $\mathrm{Spec}(C^\bullet(\fn))$---the derived analogue of the nilpotence variety.
\end{rmk}

\begin{rmk}
	In many cases the multiplet $A^\bullet(\Gamma)$ can be equipped with additional structures such as higher brackets yielding an $L_\infty$ structure. For example if $\Gamma$ is not only a $C^\bullet(\fn)$-module but in addition an algebra, one can take a Lie algebra $\fh$ such that the tensor product $\Gamma \otimes \fh$ forms a dgs Lie algebra. Then $A^\bullet(\Gamma \otimes \fh) = A^\bullet(\Gamma) \otimes \fh$ is also a dgs Lie algebra. An $L_\infty$ structure on the component fields arises via homotopy transfer. This was studied in the case of ten-dimensional super Yang--Mills theory in~\cite{AKLL, perspectives}.
\end{rmk}

\subsection[The multiplet associated to C\^{}bullet(n)]{The multiplet associated to $\boldsymbol{C^\bullet(\fn)}$} \label{sec: forms}

As a first example we can plug $C^\bullet(\fn)$ itself into the derived pure spinor functor $A^\bullet$ and study the associated multiplet.

\begin{lem}
	There is a natural equivalence
	\begin{equation*}
	A^\bullet\big(C^\bullet(\fn)\big) \simeq \Omega^\bullet(N).
	\end{equation*}
	
\end{lem}

\begin{proof}
	First, we can describe
	\begin{equation*}
	A^\bullet\big(C^\bullet(\fn)\big) = \left( C^\infty(N) \otimes C^\bullet(\fn) , \cD = \lambda^\alpha \sR(d_\alpha) + v^\mu \sR(e_\mu) + \d_{\mathrm{CE}} \right) .
	\end{equation*}
	Let us write $V$ for the even part $N_2$ of $N$, viewed as an affine space. We can identify $C^\infty(N) = C^\infty(V) \otimes \CC[\theta^\alpha]$ and $C^\bullet(\fn) = \CC[\lambda^\alpha , v^\mu]$. The differential takes the explicit form
	\begin{equation*}
	\cD = \lambda^\alpha \frac{\partial}{\partial \theta} - \lambda^\alpha f^\mu_{\alpha \beta} \theta^\beta \frac{\partial}{\partial x^\mu} + v^\mu \frac{\partial}{\partial x^\mu} + \big(\lambda^\alpha f^\mu_{\alpha \beta} \lambda^\beta\big) \frac{\partial}{\partial v^\mu} .
	\end{equation*}
	On the other hand, $N$ is parallelizable, therefore its de Rham complex takes the form
	\begin{equation*}
	\Omega^\bullet(N) \cong C^\infty(N) \otimes \Sym^\bullet(\fn_1^*) \otimes \wedge^\bullet\fn_2^* .
	\end{equation*}
	Identifying $\d\theta^\alpha = \lambda^\alpha$ and $\d x^\mu = v^\mu$, the de Rham differential takes the form
	\begin{equation*}
	\d_{\mathrm{dR}} = \lambda^\alpha \frac{\partial}{\partial \theta^\alpha} + v^\mu \frac{\partial}{\partial x^\mu} .
	\end{equation*}
	Further, $\fg$ acts on the de Rham complex via left translation making it a strict multiplet. The map defined in coordinates via
	\begin{equation*}
	\big( A^\bullet\big(C^\bullet(\fn)\big) , \cD , \sL \big) \longrightarrow \left( \Omega^\bullet(N) , \d_{\mathrm{dR}} , \sL \right), \qquad \big(x^\mu, \lambda^\alpha, \theta^\alpha , v^\mu\big) \mapsto \big(x^\mu, \lambda^\alpha, \theta^\alpha, v^\mu + \lambda f^\mu \theta\big)
	\end{equation*}
	is a quasi-isomorphism of $\fg$-multiplets. Therefore, we can identify the multiplet associated to $C^\bullet(\fn)$ itself as the differential forms on the super Lie group $N$.
\end{proof}

\begin{rmk}
	We can further compute cohomology with respect to $\cD_0 = \lambda^\alpha \frac{\partial}{\partial \theta^\alpha}$ and obtain a~retraction\footnote{Recall that a retraction between two cochain complexes is a diagram as in~\eqref{eq: dR transfer}, such that $p \circ i = \id_{\Omega^\bullet(V)}$ and $i \circ p - \id_{\Omega^\bullet(N)} = \cD_0 \circ h + h \circ \cD_0$.} to the de Rham complex on the spacetime manifold $V$:
	\begin{equation} \label{eq: dR transfer}
	\begin{tikzcd}
	\arrow[loop left]{l}{h}\big(\Omega^\bullet(N) , \cD_0\big)\arrow[r, shift left, "p"] & \big(\Omega^\bullet(V) , 0\big).\arrow[l, shift left, "i"]
	\end{tikzcd}
	\end{equation}
	Here, $i$ is the embedding of the factor of polynomial degree 0 in the $\lambda$ and $\theta$ variables, and $p$ is the obvious projection. The homotopy $h$ is given by $h=\theta^\alpha \frac{\partial}{\partial \lambda^\alpha}$.
	The induced differential via homotopy transfer is the de Rham differential on $V$. The module structure, however, is no longer strict. In fact the strict part now vanishes, but there are now quadratic pieces appearing. These take the form
	\begin{equation*}
	\rho^{(2)} (Q , Q) = \iota_{[Q, Q ]} \colon \ \Omega^k(V) \longrightarrow \Omega^{k-1}(X) .
	\end{equation*}
	Here we view $[Q, Q]$ as a constant vector field on $V$ and $\iota$ denotes the contraction of a differential form with a vector field.

	There is of course a further quasi-isomorphism of $\fg$-modules to the trivial $\fg$-module $\CC$. In this spirit, each of the multiplets $A^\bullet\big(C^\bullet(\fn)\big)$, $\Omega^\bullet(N)$, and $\Omega^\bullet(V)$ can be viewed as free resolutions of the trivial module---either free over spacetime (i.e., as $C^\infty(V)$-modules), or even free over superspace (i.e., as $C^\infty(N)$-modules). Note, however, that the trivial module does not form a $\fg$-multiplet since the translations do not act geometrically. Therefore this last quasi-isomorphism is just a quasi-isomorphism of $\fg$-modules.
\end{rmk}

\subsection{Component fields and homotopy transfer} \label{sec: components}
The multiplet $A^\bullet(\Gamma)$ given by applying the pure spinor functor does not at first glance resemble the more standard component field formulations known from physics. These component field multiplets are distinguished by the fact that they are given by dgs vector bundles whose total rank (as vector bundles over spacetime) is finite. In other words, component field multiplets are usually defined with the assumption that they only contain a finite number of component fields. With some care, it is always possible to choose a homotopy representative for $A^\bullet(\Gamma)$ that satisfies this condition, as long as $\Gamma$ satisfies a finiteness condition. In other words, we will establish the following lemma.

\begin{lem} \label{fd_model_lemma}
	Let $\Gamma$ be a $C^\bullet(\fn)$-module such that only finitely many cohomology groups $\mathrm{H}^\bullet(\Gamma)$ are non-vanishing and each of these is finitely generated as an $R$-module. Then $A^\bullet(\Gamma)$ is quasi-isomorphic to a multiplet of finite rank.
\end{lem}

In the rest of this section we will explain an algorithm that provides explicit finitely generated component field representatives for $A^\bullet(\Gamma)$, thus establishing the lemma. In~\cite{perspectives}, a general procedure to extract such a component field multiplet out of the underived model $A^\bullet_{R/I}(\Gamma)$ was explained. In essence, this is done by homotopy transfer: one identifies a retraction to another quasi-isomorphic complex and then applies homotopy transfer to the structures present for the multiplet. By construction, this yields a quasi-isomorphic and thus physically equivalent multiplet. Crucially, the component field multiplets obtained in that way are not necessarily strict anymore. Higher-order terms in the module structure can arise during the transfer.

Let us begin by summarizing the procedure for the underived pure spinor formalism $A^\bullet_{R/I}$ and then describe the generalization to $A^\bullet$.

\subsubsection[Minimal models for A\^{}bullet\_\{R/I\}]{Minimal models for $\boldsymbol{A^\bullet_{R/I}}$}

Recall that the differential on $\big(A^\bullet_{R/I} (\Gamma),\cD\big)$ admits an obvious splitting
\begin{equation*}
\cD = \cD_0 + \cD_1 = \lambda^\alpha \frac{\partial}{\partial \theta^\alpha} - \lambda^\alpha f^\mu_{\alpha \beta} \theta^\beta \frac{\partial}{\partial x^\mu} ,
\end{equation*}
which can be viewed by equipping $A^\bullet_{R/I}(\Gamma)$ with a filtration by polynomial degree on $V$ and splitting the differential into the terms that preserve and lower filtered degree.
(To filter by polynomial degree on~$V$ is imprecise, since that filtration is not complete. A complete filtration that serves the same purpose can be constructed as explained in~\cite{perspectives}, roughly speaking by assigning both $\lambda$ and $\theta$ filtered degree one. Alternatively, under the assumption that the generator of scalings is included in~$\fg_0$, one can use the weight grading as discussed in the next subsection.)
We can take cohomology with respect to $\cD_0$ and then perform homotopy transfer along the diagram
\begin{equation*}
\begin{tikzcd}
\arrow[loop left]{l}{h}(A^\bullet (\Gamma) , \cD_0)\arrow[r, shift left, "p"] &(\mathrm H^\bullet(A^\bullet(\Gamma), \cD_0) , \, 0)\arrow[l, shift left, "i"] .
\end{tikzcd}
\end{equation*}
We observe, by \cite[Theorem 10.3.15]{LodayVallette}, that although the transfer data depends on a choice of section for the projection onto the cohomology, the multiplet obtained by homotopy transfer is independent of this choice up to isomorphism. In this example, it is always ``minimal'' in the sense that its differential does not contain terms of order zero in differential operators.
Intuitively, this means that we cannot take any further cohomology without leaving the category of multiplets. (In more general examples coming from $C^\bu(\fn)$-modules, ``minimal'' is not related to having no terms of order zero in the differential; the massive Klein--Gordon field is a simple example.)

Starting from a multiplet of the form $A^\bullet_{R/I}(\Gamma)$, we will refer to the corresponding minimal multiplet as $\mu A^\bullet_{R/I}(\Gamma)$.
One can identify the $\cD_0$-cohomology with the Koszul cohomology of $\Gamma$, tensored with functions on spacetime,
\begin{equation*}
\mathrm H^\bullet(A_{R/I}^\bullet(\Gamma)) = C^\infty(V) \otimes \mathrm H^\bullet(K^\bullet(\Gamma)) .
\end{equation*}
The Koszul cohomology is conveniently computed by a minimal free resolution $L$ of $\Gamma$ in $R$-modules. The minimal multiplet takes the form
\begin{equation*}
\mu A^\bullet_{R/I}(\Gamma) = \big( C^\infty(V) \otimes (L \otimes_R \CC) , \cD' ,\rho' \big) ,
\end{equation*}
where $\cD'$ is the differential induced from $\cD_1$ and $\rho'$ the module structure induced from $\sL$ via homotopy transfer. We refer to~\cite{perspectives} for details. For $\Gamma$ a finitely generated module, Hilbert's syzygy theorem states that the minimal free resolution exists, consists of finitely generated modules and its length is less or equal than $\dim(\fn_1)$ (see for example~\cite[Theorem 1.13]{MR1322960}; for a~discussion in the equivariant case we refer to~\cite[Proposition 2.4.9, Remark 2.4.10]{Galetto}). Therefore,~$\mu A_{R/I}^\bullet(\Gamma)$ is indeed of finite rank over spacetime, i.e., it provides a reasonable component field multiplet.

\subsubsection[Minimal models for A\^{}bullet]{Minimal models for $\boldsymbol{A^\bullet}$}

Let us now discuss a generalization of this procedure to the functor $A^\bullet$. We will assume in this section that $\Gamma$ carries an action of the abelian Lie algebra $\RR$ compatible with the action of $\RR$ on $\fn$ where $\fn_i$ has $\RR$-weight $i$. This will be used to construct a filtration at the very end of this section (however, this assumption is not required for Lemma~\ref{fd_model_lemma}).

Recall that the differential on $A^\bullet(\Gamma)$ takes the form
\begin{equation*}
\cD = \lambda^\alpha \frac{\partial}{\partial \theta^\alpha} - \lambda^\alpha f^\mu_{\alpha \beta} \theta^\beta \frac{\partial}{\partial x^\mu} + \d_\Gamma + v^\mu
\frac{\partial}{\partial x^\mu} .
\end{equation*}
One can construct a finite-dimensional component field model in two steps. First one takes cohomology with respect to the internal differential of the module $\d_\Gamma$ and performs homotopy transfer along the diagram
\begin{equation*}
\begin{tikzcd}
\arrow[loop left]{l}{h}(A^\bullet (\Gamma) , \d_\Gamma)\arrow[r, shift left, "p"] &(C^\infty(N) \otimes \mathrm H^\bullet(\Gamma) ) , \, 0)\arrow[l, shift left, "i"] .
\end{tikzcd}
\end{equation*}
The differential and the $C^\bullet(\fn)$-module structure of the resulting multiplet may contain additional pieces induced from homotopy transfer. Since $\mathrm H^\bullet(\Gamma)$ is a $C^\bullet(\fn)$-module with vanishing differential, each homogeneous summand $\mathrm H^k(\Gamma)$ carries the structure of an $R/I$-module. Thus, we can proceed by taking cohomology with respect to the Koszul differential $\cD_0 = \lambda^\alpha \frac{\partial}{\partial \theta^\alpha}$, and applying the homotopy transfer along a retraction on to this cohomology. As before, the $\cD_0$-cohomology is computed by minimal free resolutions of the individual $\mathrm H^k(\Gamma)$ in $R$-modules, and as before, it is independent of the choice of transfer datum. The resulting multiplet is thus of the form
\begin{equation*}
C^\infty(V) \otimes \bigg(\bigoplus_k L^\bullet_k \otimes_R \CC\bigg) ,
\end{equation*}
where $L^\bullet_k$ is the minimal free resolution of $\mathrm H^k(\Gamma)$. As long as the cohomology $\mathrm H^\bullet(\Gamma)$ is bounded, this is already a finite rank vector bundle over spacetime. The field content of this multiplet is just the field content of the direct sum of all the minimal multiplets associated to the cohomology groups, i.e.,
\begin{equation*}
\bigoplus_k \mu A^\bu\big(\mathrm H^k(\Gamma)\big) .
\end{equation*}
At this stage we have already obtained a finite-dimensional model, as required by Lemma \ref{fd_model_lemma}. However, additional acyclic differentials induced by homotopy transfer can still be present. At this final stage, we will use the filtration on each multiplet $\mu A^\bu\big(\mathrm H^k(\Gamma)\big)$ by weight with respect to the $\RR$-action on $\Gamma$. We define $\mu A^\bullet(\Gamma)$ by taking the cohomology with respect to the summand of the total differential of filtered degree zero (in other words, we pass to the $E_1$ page of the associated spectral sequence), and apply homotopy transfer. We will illustrate this procedures using examples in~Section~\ref{sec: examples}.

\section{An equivalence of categories} \label{sec: equiv}
We now show that the derived pure spinor superfield formalism provides an equivalence of categories between the dg categories of $\fg_0$-equivariant $C^\bullet(\fn)$-modules and $\fg$-multiplets. This implies in particular that, up to quasi-isomorphism, every $\fg$-multiplet can be constructed using the derived pure spinor superfield formalism.

\begin{rmk}
	Recall that for a general Lie algebra there is a Koszul duality equivalence between the dg-categories of $C_\bullet(\fg)$-comodules and $U(\fg)$-modules. If $\fg$ is, for instance, finite dimensional, we can instead consider $C^\bullet(\fg)$-modules. For a relevant discussion in a similar context, see~\mbox{\cite[Sections~7 and~8]{CostelloYangian}}. Relatedly, Kapranov~\cite{Kapranov} gives a formulation of Koszul duality that establishes a Quillen equivalence between the categories of $D$-modules and $\Omega^\bu$-modules on the same space. Our results show that the pure spinor formalism admits a natural derived generalization that is closely related to Kapranov's construction; however, in our setting, one is working on a supermanifold, and asks for appropriate equivariance conditions. Alternatively, our procedure can be viewed in two steps, as an explicit form of commutative/Lie Koszul duality tailored to the examples in question, combined with an associated bundle construction that realizes a~$U(\fn)$-module as an $\fn$-equivariant dg vector bundle over $V = N_2$.
\end{rmk}

\subsection[The inverse functor: derived n-invariants]{The inverse functor: derived $\boldsymbol{\fn}$-invariants}

Any $\fg$-module is in particular a $\fn$-module, and we can thus take derived invariants with respect to $\fn$. For multiplets, this defines a functor in the opposite direction to the functor $A^\bullet$, assigning a strict $C^\bullet(\fn)$-module to a multiplet. The resulting $C^\bullet(\fn)$-module is also $\fg_0$-equivariant.
The functor takes the form
\begin{equation*}
C^\bullet(\fn , -) \colon \ \Mult_{\fg}^{\mathrm{strict}} \longrightarrow \Mod^{\fg_0}_{C^\bullet(\fn)},
\end{equation*}
It maps the multiplet	$(E,D,\rho)$ to $C^\bullet(\fn , \cE)$; this Chevalley--Eilenberg complex can be written more explicitly as
\begin{equation*}
C^\bullet(\fn , \cE) = \big( C^\bullet(\fn) \otimes \cE , \d_{\mathrm{CE}} + D + \lambda^\alpha \rho(d_\alpha) + v^\mu \rho(e_\mu) \big) .
\end{equation*}
It is a strict $C^\bullet(\fn)$-module; the action is on $C^\bullet(\fn)$ via multiplication and on $\cE$ via the identity. Morphisms are mapped according to the rule
\begin{equation*}
\psi \mapsto
\id_{C^\bullet(\fn)} \otimes \psi .
\end{equation*}

\begin{rmk}
	The assignment on the level of objects is also well defined for not necessarily strict modules. Then the higher-order terms of the $\fg$-module structure enter the differential such that the term $\lambda^\alpha \rho(d_\alpha)$ is replaced by
	\begin{equation*}
	\lambda \cdot \rho := \sum_{k=1}^{\infty} \lambda^{\alpha_1} \cdots \lambda^{\alpha_k} \rho^{(k)}(d_{\alpha_1} , \dots , d_{\alpha_k}).
	\end{equation*}
	Notice that the output is still a strict $C^\bullet(\fn)$-module.
\end{rmk}
There are several ways to intuitively understand the fact that the inverse functor is given by the derived invariants of supertranslations. One can of course appeal to the general structure of Koszul duality. On a more down-to-earth level, one can recall that the component fields of a supermultiplet in the usual pure spinor superfield formalism correspond to the generators of the minimal free resolution of that module over $R = C^\bu(\Pi \fn_1)$. The resolution differential was then proven in~\cite{perspectives} to agree with those supersymmetry transformations that are order zero in spacetime derivatives, verifying a conjecture of Berkovits. Since spacetime derivatives act by zero on translation-invariant sections of~$E$, it is clear that one can think of the minimal free resolution of the module as arising from the derived $\Pi \fn_1$-invariants of translation-invariant sections: $R$ consists of the Lie algebra cochains of~$\Pi \fn_1$, the generators of the free graded $R$-module arise from the component fields of the multiplet, and the differential encodes the action of supersymmetry on translation-invariant sections. It is clear that this story is just a two-step procedure to compute derived $\fn$-invariants by first taking $\fn_2$-invariants, and then accounting for the action of supersymmetry.

To flesh this story out, we now describe the module $C^\bullet(\fn , \cE)$ in some more detail and sketch how its cohomology can be computed. Recall that, since $\cE$ is a multiplet, the action of $\fn_2$ is just given by derivatives along the coordinate directions
\begin{equation*}
\rho(e_\mu) = \frac{\partial}{\partial x^\mu} .
\end{equation*}
Thus taking cohomology with respect to the term $v^\mu \rho(e_\mu)$ in the differential means restricting to translation invariant sections of $\cE$ and eliminating $v^\mu$. Denoting the fiber of $E$ over 0 by $E_0$ we find a quasi-isomorphism
\begin{equation} \label{eq: compute C^bu}
C^\bullet(\fn ,\cE) \simeq \big( R \otimes E_0^\bullet , \lambda \cdot \rho_{\mathrm{constants}} \big) .
\end{equation}
Note that $E_0^\bullet$ carries a $\ZZ \times \ZZ/2\ZZ$-grading since it comes from a dgs vector bundle. In addition, there is an integer grading by polynomial degree in $\lambda$ present. In the following, we will label cohomology groups by the former degree. Then each cohomology group is an $\fg_0$-equivariant~$R/I$-module graded by polynomial degree in $\lambda$. If the cohomology is concentrated in a single degree, the complex on the right hand-side can be viewed as the minimal free resolution of the $R/I$-module forming the cohomology. Note that this is indeed a minimal free resolution, since all terms in the differential are of nonzero polynomial degree in $\lambda$. It will follow from the theorem below that this $R/I$-module is precisely the algebraic input datum the multiplet can be constructed from in the pure spinor superfield formalism, i.e., by applying the functor $A^\bullet_{R/I}$.

\subsection{Main theorem and proof}
Let us now show that the functors $A^\bullet$ and $C^\bullet(\fn,-)$ induce an equivalence of dg-categories between the homotopy categories of multiplets and equivariant $C^\bullet({\fn})$-modules.

\begin{thm} \label{thm: equiv}
	$A^\bullet$ and $C^\bullet(\fn , - )$ provide an equivalence of dg-categories between $\Mult_\fg^{\text{\rm strict-ob}}$ and $\Mod^{\fg_0}_{C^\bullet({\fn})}$.
\end{thm}
\begin{proof}
	We first show that there is an equivalence of dg-functors
	\begin{equation*}
	\mathrm{id}_{\Mod^{\fg_0}_{C^\bullet({\fn})}} \to C^\bullet \circ A^\bullet .
	\end{equation*}
	We will naturally construct quasi-isomorphisms
	\begin{equation*}
	\Gamma \simeq C^\bu(\fn,A^\bu(\Gamma))
	\end{equation*}
	for each equivariant $C^\bu(\fn)$-module $(\Gamma,\d_\Gamma)$.
	
	We can explicitly describe $C^\bu(\fn,A^\bu(\Gamma))$ by
	\begin{gather*}
	\bigg( \CC\big[\lambda', v'\big] \otimes C^\infty(N) \otimes \Gamma ,\\
\qquad{}  \d_\Gamma + \big(\lambda' f^\mu \lambda'\big) \frac{\partial}{\partial v'^\mu} + v'^\mu \frac{\partial}{\partial x^\mu} + \lambda'^\alpha \sL(d_\alpha) + \lambda^\alpha \sR(d_\alpha) + v^\mu \sR(e_\mu) \bigg) .
	\end{gather*}
	Here we made a notational distinction between the generators of the $C^\bullet(\fn)$ in the construction (denoted by $\lambda'$ and $v'$) and the action of $C^\bullet(\fn)$ on $\Gamma$ (denoted by $\lambda$ and $v$).
	
	The differential contains a piece of the form
	\begin{equation*}
	\d_{\mathrm{dR}} = v'^\mu \frac{\partial}{\partial x^\mu} + \lambda'^\alpha \frac{\partial}{\partial \theta^\alpha} .
	\end{equation*}
	Thus, we can identify
	\begin{equation*}
	\left( C^\bu(\fn,A^\bu(\Gamma)) , v'^\mu \frac{\partial}{\partial x^\mu} + \lambda'^\alpha \frac{\partial}{\partial \theta^\alpha} \right) = \left( \Omega^\bu(N) , \d_{\mathrm{dR}} \right) \otimes \Gamma .
	\end{equation*}
	Since the de Rham complex on $N$ is acyclic, this complex is quasi-isomorphic to $\Gamma$. We can fix homotopy data
	\begin{equation*}
	\begin{tikzcd}
	\arrow[loop left]{l}{h} \left( \Omega^\bu(N) \otimes \Gamma , \d_{\mathrm{dR}} \right) \arrow[r, shift left, "p"] & (\Gamma , \, 0)\arrow[l, shift left, "i"] .
	\end{tikzcd}
	\end{equation*}
	It is easy to see that the only induced differential on the right hand side is $\d_\Gamma$, so that we obtain a quasi-isomorphism uniformly for all choices of $\Gamma$:
	\begin{equation*}
	C^\bu(\fn,A^\bu(\Gamma)) \simeq (\Gamma , \d_\Gamma) .
	\end{equation*}
	
	Now, for any morphism $f \colon \ \Gamma \to \Gamma'$ of $C^\bullet(\fn)$-modules, we can identify
	\begin{equation*}
	C^\bu(A^\bu(f)) = \id_{C^\bu(\fn)} \otimes \id_{C^\infty(N)} \otimes f.
	\end{equation*}
	So there is a commutative square
	\begin{equation*}
	\begin{tikzcd}
	\Gamma \arrow[d , "f"] \arrow[r] & C^\bullet(\fn , A^\bullet(\Gamma)) \arrow[d, "A^\bullet(f)"] \\
	\Gamma' \arrow[r] &C^\bullet(\fn , A^\bullet(\Gamma'))
	\end{tikzcd}
	\end{equation*}
	inducing an equivalence of hom complexes $\mathrm{Hom}(\Gamma, \Gamma') \to \mathrm{Hom}(C^\bullet(\fn , A^\bullet(\Gamma)), C^\bullet(\fn , A^\bullet(\Gamma')))$.

	Conversely, we will construct an equivalence of dg-functors
	\begin{equation*}
	A^\bullet \circ C^\bullet \to \mathrm{id}_{\Mult_\fg^{\text{\rm strict-ob}}} .
	\end{equation*}
	
	Let $(E,D,\rho)$ be a $\fg$-multiplet. We can describe $A^\bullet(C^\bu(\fn , \cE))$ explicitly by
	\begin{equation*}
	\left( \CC[\lambda, v] \otimes \cE \otimes C^\infty(N) , D + \lambda f^\mu \lambda \frac{\partial}{\partial v^\mu} + \lambda \cdot \rho + v^\mu \rho(e_\mu) + \lambda^\alpha \sR(d_\alpha) + v^\mu \sR(e_\mu) \right) .
	\end{equation*}
	We denote coordinates on the base of the vector bundle $E$ by $x^\mu$ and on $N_2$ by $y^\mu$. In this notation, we find
	\begin{equation*}
	v^\mu\left( \rho(e_\mu) + \sR(e_\mu) \right) = v^\mu \left( \frac{\partial}{\partial x^\mu} + \frac{\partial}{\partial y^\mu} \right) .
	\end{equation*}
	Taking cohomology with respect to this piece, we obtain a quasi-isomorphic complex of the form
	\begin{equation*}
	\big( \CC[\lambda, \theta] \otimes \cE , D + \lambda \cdot \rho + \lambda^\alpha \sR(d_\alpha) \big) .
	\end{equation*}
	The differential contains a piece corresponding to the Koszul differential
	\begin{equation*}
	\d_K = \cD_0 = \lambda^\alpha \frac{\partial}{\partial \theta^\alpha} .
	\end{equation*}
	Clearly, taking cohomology with respect to the Koszul differential we arrive at $\cE$. Concretely let us consider homotopy data
	\begin{equation*}
	\begin{tikzcd}
	\arrow[loop left]{l}{\cD_0^\dagger}(\CC[\lambda,\theta] \otimes \cE , \cD_0 )\arrow[r, shift left, "p"] &( \cE , \, 0)\arrow[l, shift left, "i"] ,
	\end{tikzcd}
	\end{equation*}
	where the homotopy is given by $\cD_0^\dagger = \theta \frac{\partial}{\partial \lambda}$, while $i$ is the obvious inclusion and $p$ evaluates at $\lambda = \theta = 0$. It is easy to see that the induced differential is just $D$. Thus, we find homotopy data
	\begin{equation} \label{eq: hdata}
	\begin{tikzcd}
	\arrow[loop left]{l}{\cD_0'^\dagger}(\CC[\lambda,\theta] \otimes \cE , d )\arrow[r, shift left, "p'"] &( \cE , \, D )\arrow[l, shift left, "i'"] ,
	\end{tikzcd}
	\end{equation}
	providing a quasi-isomorphism of cochain complexes. The induced maps are given by
	\begin{align}
	&i' = \sum_{n=0}^{\infty} \big(\cD_0^\dagger(D + \cD_1 + \lambda \cdot \rho)\big)^n \circ i,\nonumber \\
	&p' = p \circ \sum_{n=0}^{\infty} \big(\cD_0^\dagger(D + \cD_1 + \lambda \cdot \rho)\big)^n = p,\nonumber \\
	&\cD_0'^\dagger = \cD_0^\dagger \circ \sum_{n = 0}^{\infty} \big((D + \cD_1 + \lambda \cdot \rho) \cD_0^\dagger\big)^n.	\label{htpy_transfer_in_main_thm}
	\end{align}
	By construction, $p \circ \cD_0^\dagger = 0$ and thus $p' = p$. We note that these sums are all finite by degree reasons, since the left hand side in~\eqref{eq: hdata} is concentrated in finitely many degrees with respect to the grading by polynomial degree in the $\theta$-variables, and in each expression in equation~\eqref{htpy_transfer_in_main_thm}, the operator being raised to the power $n$ within the sum raises this $\theta$-degree. Therefore for $n$ sufficiently large, all terms in the sum defining our maps vanish.
	
	We have to check that this not only provides a quasi-isomorphism of cochain complexes, but of multiplets. Therefore we transfer the module structure induced by $\sL$ to the right hand side and check that it agrees with the original module structure $\rho$ of the multiplet.
	
	Explicitly, the transferred module structure can be described by
	\begin{align*}
	\rho_\sL^{(k)} (Q,\dots,Q)={}& p \sL(Q) \big(\cD_0'^\dagger \sL(Q)\big)^{k-1} i' \\
	={}&p \sL(Q) \left(\cD_0^\dagger \sum_{n = 0}^{\infty} ((D + \cD_1 + \lambda \cdot \rho) \cD_0^\dagger)^n \sL(Q)\right)^{k-1}\\
&{}\times \sum_{j=0}^{\infty} \big(\cD_0^\dagger(D + \cD_1 + \lambda \cdot \rho)\big)^j \circ i .
	\end{align*}
	Since $p$ projects onto $\lambda = \theta = 0$, only terms of order zero in $\lambda$ and $\theta$ can contribute. It is easy to see that the only such term is
	\begin{equation*}
	\rho_\sL^{(k)} (Q,\dots,Q) = p \epsilon \frac{\partial}{\partial \theta} \left(\cD_0^\dagger \epsilon \frac{\partial}{\partial \theta}\right)^{k-1} \cD_0^\dagger \lambda \cdot \rho^{(k)} i .
	\end{equation*}
	Here we expressed $Q$ in a basis, $Q = \epsilon^\alpha d_\alpha$.
	Noting that
	\begin{equation*}
	\left\{\cD_0^\dagger , \epsilon \frac{\partial}{\partial \theta}\right\} = \epsilon \frac{\partial}{\partial \lambda} ,
	\end{equation*}
	we deduce
	\begin{equation*}
	\rho_\sL^{(k)} (Q,\dots,Q) = \rho^{(k)} (Q,\dots,Q) .
	\end{equation*}
	This shows that $A^\bu(C^\bu(\fn , \cE)) \simeq (E,D,\rho)$ as multiplets.
	
	As before, for any morphism $\psi \colon \ \cE \to \cE'$ we can realize
	\begin{equation*}
	A^\bu(C^\bu(\psi)) = C^\bullet(\psi) \otimes \id_{C^\infty(N)} .
	\end{equation*}
	So there is a commutative diagram
	\begin{equation*}
	\begin{tikzcd}
	A^\bullet(C^\bullet(\fn,\cE)) \arrow[r] \arrow[d, "C^\bullet(\psi)"] & \cE \arrow[d , "\psi"] \\
	A^\bullet(C^\bullet(\fn,\cE')) \arrow[r] & \cE'
	\end{tikzcd}
	\end{equation*}
	inducing an equivalence of hom spaces, and hence we have an equivalence of dg-categories.
\end{proof}

\begin{rmk}
	The equivalence of dg-categories automatically induces an equivalence of the underlying homotopy categories. Hence, each multiplet is perturbatively equivalent to a multiplet constructed via the derived pure spinor superfield formalism.
\end{rmk}

\subsection{Some consequences of the theorem}
Let us now discuss some consequences of the equivalences of categories.

\begin{cor}
	Let $(E,D,\rho)$ be any multiplet. $A^\bullet(C^\bullet(\fn,\cE))$ is a strictification.
\end{cor}
Note that this does not imply that there exist a strict finitely generated component field multiplet that is equivalent to $\cE$. In particular, $\cE$ need not admit an auxiliary field formulation in the usual sense. The strictification $A^\bullet(C^\bullet(\fn,\cE))$ typically contains infinitely many component fields.

We can further use the equivalence of categories to derive some statements on the (underived) pure spinor superfield formalism, i.e., the functor $A^\bullet_{R/I}$.

\begin{cor} \label{underived_recognition_cor}
	The essential image of the functor $A^\bullet_{R/I}$ consists of those multiplets for which $\mathrm H^\bullet(\fn , \cE)$ is concentrated in a single $\ZZ$-degree after totalization of the natural $\ZZ \times \ZZ$-grading.
\end{cor}

This gives a description of all multiplets which can be constructed via the pure spinor superfield formalism. In~\cite{perspectives}, it was argued that the antifield multiplet of the four-dimensional vector multiplet cannot be constructed via $A^\bullet_{R/I}$. We come back to this example in~Section~\ref{sec: vector anti} where we compute the relevant cohomology spaces and explain how the multiplet is built in the derived formalism.

\begin{cor}
	The functors $A^\bullet_{R/I}$ and $\mathrm H^\bullet(\fn , \cE)$ provide an equivalence of categories between the essential image of $A^\bullet_{R/I}$ and the category of graded $\fg_0$-equivariant $R/I$-modules.
\end{cor}

Suppose we have an $R/I$-module $\Gamma$ with associated minimal multiplet $\mu A_{R/I}^\bullet(\Gamma)$. As discussed earlier, the fields of $\mu A_{R/I}^\bullet(\Gamma)$ take values in the minimal free resolution of $\Gamma$. In addition, there is a close link between the supersymmetry module structure and the resolution differential. In more detail, let us pull back $\mu A_{R/I}^\bullet(\Gamma)$ along the inclusion
\begin{equation*}
\{0\} \hookrightarrow V .
\end{equation*}
This restricts the multiplet to the fiber $\mu A^\bullet(\Gamma)_0$. The multiplet $\mu A_{R/I}^\bullet(\Gamma)_0$ carries a module structure for the odd abelian super Lie algebra $\Pi \fn_1$. This module structure coincides with the resolution differential in the following sense.

\begin{cor} \label{cor: susy trafo}
	Let $\Gamma$ be an $R/I$-module and let $(L,\d_L)$ be its minimal free resolution in $R$-modules. Let us identify the fiber
	\begin{equation*}
	\mu A_{R/I}^\bullet(\Gamma)_0 = L \otimes_R \CC.
	\end{equation*}
	The map generated over $R$ by the $\Pi \fn_1$-module structure on $\mu A_{R/I}^\bullet(\Gamma)_0$
	\begin{equation*}
	\rho_{\mathrm{constants}} \colon \ \mu A^\bullet_{R/I}(\Gamma)_0 \otimes \bigg(\bigoplus_k \fn_1^{\otimes k}\bigg) \longrightarrow \mu A^\bullet_{R/I}(\Gamma)_0
	\end{equation*}
	coincides with the resolution differential. In coordinates we can express this as
	\begin{equation*}
	\lambda \cdot \rho_{\mathrm{constants}} = \sum_k \rho_{\mathrm{constants}}^{(k)} \big(\lambda^{\alpha_1} d_{\alpha_1}, \dots , \lambda^{\alpha_k} d_{\alpha_k}\big) = \d_L,
	\end{equation*}
	where $d_\alpha$ is a basis for $\fn_1$.
\end{cor}

\begin{proof}
	By construction, we have $\mathrm H^\bullet(\fn , \mu A^\bu(\Gamma)) = \Gamma$. But by~\eqref{eq: compute C^bu} we know that there is a~quasi-isomorphism
	\begin{equation*}
	C^\bullet(\fn , \mu A^\bu(\Gamma)) \simeq \left( \mu A^\bu(\Gamma)_0 \otimes R , \lambda \cdot \rho_{\mathrm{constants}} \right) .
	\end{equation*}
	Thus, the cochain complex on the right is the minimal free resolution of $\Gamma$ in $R$-modules and we obtain the desired result.
\end{proof}

In practice this means that the resolution differential contains all the information on the supersymmetry transformations which are of order zero in the derivatives. This result was conjectured by Berkovits in~\cite{BerkovitsSupermembrane} and proved by direct computation in~\cite{perspectives}.

Let us now state some results on the duality operations in the category of multiplets.

\begin{cor}
	Let $(E,D,\rho)$ be a $\fg$-multiplet and let $(E^*, D^*, \rho^*)$ be the respective dual $($or antifield$)$ multiplet. If these are both quasi-isomorphic to a multiplet in the image of $A^\bullet_{R/I}$ and we have
	\begin{equation*}
	(E,D,\rho) \simeq A^\bullet_{R/I}(\Gamma) ,
	\end{equation*}
	then there is also a quasi-isomorphism
	\begin{equation*}
	\big(E^*, D^*,\rho^*\big) \simeq A^\bullet_{R/I}\big(\Ext_R^{n-q}(\Gamma,R)\big) .
	\end{equation*}
	Here, $n= \dim(\fn_1)$ and $q = \dim_R(\Gamma)$.
\end{cor}
\begin{proof}
	Move to the minimal multiplet $\mu A^\bullet_{R/I}(\Gamma) \simeq (E,D,\rho)$. Corollary~\ref{cor: susy trafo} implies that
	\begin{equation*}
	C^\bullet(\fn , \cE) \simeq C^\bullet(\fn , \mu A^\bullet_{R/I}(\Gamma)) \simeq (L,\d_L) ,
	\end{equation*}
	where $(L,d_L)$ is the minimal free resolution of $\Gamma$ in $R$-modules. There is a quasi-isomorphism
	\begin{equation*}
	\big(E^*,D^* ,\rho^*\big) \simeq \big(\mu A^\bullet_{R/I}(\Gamma) \big)^* .
	\end{equation*}
	This in turn implies that
	\begin{equation*}
	C^\bullet\big(\fn , \cE^*\big) \simeq C^\bullet\big(\fn , \mu A^\bullet_{R/I}(\Gamma)^* \big) \simeq \big(L^* , \d_L^*\big) ,
	\end{equation*}
	where $(L^*,\d_L^*)$ is the dual of the minimal free resolution $(L,\d_L)$. By assumption the derived invariants of $(E^*,D^*,\rho^*)$ are concentrated in a~single degree. Hence, $(L^*,\d_L^*)$ again resolves a~single $R/I$-module, namely $\Ext_R^{n-q}(\Gamma,R)$.
\end{proof}

In general, however, taking dual multiplets can lead outside of the image of $A^\bullet_{R/I}$. In fact, the above proof implies that $(A^\bullet_{R/I}(\Gamma))^*$ is in the essential image of $A^\bullet_{R/I}$ precisely when $\Gamma$ is a Cohen--Macaulay $R$-module. In general this may not be the case---the dual of the vector multiplet in four-dimensional $\cN=1$ supersymmetry is an example of this type, that does not occur in the image of the underived pure spinor functor. We will discuss this example in the generalized formalism in~Section~\ref{sec: vector anti}.

If $\Gamma$ is not Cohen--Macaulay, we can still compute the derived invariants of the associated multiplet to find
\begin{equation*}
C^\bullet(\fn, A^\bullet_{R/I}(\Gamma)^*) \simeq \mathrm{RHom}_R(\Gamma , R) \simeq \Ext^\bullet_R(\Gamma, R) .
\end{equation*}
We can therefore deduce the following natural description for the \emph{dual} of a multiplet obtained using the (underived) pure spinor formalism.

\begin{cor}
	Let $\Gamma$ be an $R/I$-module and $\mu A_{R/I}^\bu(\Gamma)$ the associated component field multiplet. Then there is a $C^\bullet(\fn)$-module structure on the Ext-algebra $\mathrm{Ext}^\bullet_R(\Gamma , R)$ such that
	\begin{equation*}
	\mu A_{R/I}^\bu(\Gamma)^* \cong \mu A^\bu(\mathrm{Ext}^\bullet_R(\Gamma , R)) .
	\end{equation*}
\end{cor}

\subsection{First examples}

Let us discuss the pure spinor formalism as an equivalence of categories in a few simple examples.

\begin{eg}
	First, let us consider the situation where there is no supersymmetry at all. Let~$V$ be an $n$-dimensional vector space. Let $\fg$ denote its Poincar\'e Lie algebra of infinitesimal isometries, and let $\fn$ denote its Lie algebra of translations, viewed as a dg Lie algebra concentrated in degree 2. So $C^\bullet(\fn) \cong \Sym^\bullet(\fn^*[1])$ is an exterior algebra in $n$ generators $v_1, \ldots, v_n$ of degree $-1$. We are, therefore, interested in $\fg_0 \cong \mathfrak{so}(n)$-equivariant dg-modules over the exterior algebra. We restrict attention to those modules with finite-dimensional cohomology in each degree.
	
	On the other hand, the dg-category $\Mult_{\fg}$ of multiplets is nothing but the dg-category of Poincar\'e equivariant dg-vector bundles on the affine space $V$: the $\fg$ action is completely determined (and all multiplets are automatically strict). Again, let us restrict attention to bundles with finite-dimensional cohomology in each degree. Such multiplets are determined by their restriction to a formal neighborhood of the origin, which is a Poincar\'e equivariant dg-module over the completed Weyl algebra $\hat D_n = \CC[\partial_1, \ldots, \partial_n][\![z_1, \ldots, z_n]\!]/([\partial_i, z_i]-1)$. Translation equivariance guarantees that all such modules are induced from $\mathfrak{so}(n)$-equivariant modules over $\CC[\partial_1, \ldots, \partial_n]$. Our statement then reduces to ordinary Koszul duality as a relationship between modules over an exterior and a commutative algebra.
\end{eg}

\begin{eg}
	Let $\Sigma$ be a finite-dimensional vector space. Let us now briefly discuss the example where $\fg = \fn = \RR \oplus \Pi \Sigma$ as a graded vector space, with Lie bracket given by a non-degenerate inner product $\Sigma^{\otimes 2} \to \RR$. This is the background for \emph{supersymmetric classical mechanics}. Now
	\begin{equation*}
	C^\bullet(\fn) \cong \big(\CC\big[v, \lambda^1, \ldots, \lambda^\cN\big] , \d_{\mathrm{CE}}\big),
	\end{equation*}
	where $v$ is an odd generator, $\lambda^i$ are even generators, $\d_{\mathrm{CE}} \lambda^i = 0$ for all $i$, and $\d_{\mathrm{CE}} v = \big(\lambda^1\big)^2 + \cdots + \big(\lambda^\cN\big)^2$ (so the $\lambda^i$ are linearly dual to a choice of orthonormal basis of $\Sigma$).
	
	Let $(\Gamma, \d_\Gamma)$ be a $C^\bullet(\fn)$-module. The derived version of the pure spinor formalism associates to $\Gamma$ the multiplet
	\begin{equation} \label{SUSYQM_multiplet}
	A^\bullet(\Gamma) = \left(C^\infty(\RR_t) \otimes \CC[\theta_1, \ldots, \theta_\cN] \otimes \Gamma , \d_\Gamma + \big(v - \lambda^a\theta_a\big)\frac{\partial}{\partial t} + \lambda^a \frac{\partial}{\partial \theta^a}\right),
	\end{equation}
	where $\theta^1, \ldots, \theta^\cN$ are, again, odd generators, and the expression $\lambda^a \theta_a$ implicitly uses the choice of inner product on $\Sigma$.
	
	On the other hand, we can study multiplets on $\RR$ for the super Lie algebra $\fn$ directly. These are affine dgs-vector bundles on $\RR$ equipped with $\cN$ commuting odd symmetries, each of which squares to the action of translation. One can see this structure more concretely from the expression of equation \eqref{SUSYQM_multiplet} by applying homotopy transfer to take the cohomology with respect to the final summand $\lambda^a \frac{\partial}{\partial \theta^a}$ of the differential.
	
	There is a rich theory underlying the classification of multiplets in supersymmetric mechanics; see for example~\cite{FauxGates}. It would be interesting to investigate the connections between this existing work and the point of view described here.
\end{eg}

\section{Applications}
\label{sec: examples}

In this final section we discuss some examples of multiplets in the light of the derived formalism. This also provides some insight to some of the curiosities of the underived pure spinor superfield formalism. In particular, we connect the multiplets associated to different Lie algebra cohomology groups by curvature maps and construct antifield multiplets (in the sense of Definition~\ref{def: antifield}) for four-dimensional $\cN=1$ supersymmetry. In addition, we show how both the on- and off-shell version of the chiral multiplet can be constructed.

\subsection{Multiplets from Lie algebra cohomology} \label{multiplet_from_cohomology_section}

We noticed in~Section~\ref{sec: forms} that the multiplet associated to $C^\bullet(\fn)$ can be identified with the de Rham complex of the super Lie group $N$. On the other hand, it was already appreciated in the previous literature that the individual Lie algebra cohomology groups of $\fn$ give an interesting class of $R/I$-modules which yield multiplets via the underived pure spinor construction $A^\bullet_{R/I}$. Some of these multiplets were studied for various super Poincar\'e algebras in~\cite[Section~6.3]{perspectives}; see also~\cite{MSX1}. The derived formalism gives a new tool to study these multiplets and to highlight the relations among the multiplets associated to the different Lie algebra cohomology groups.

To start with, recall that when the nilpotence variety $\Spec (R/I)$ associated to $\fn$ is a complete intersection, the Chevalley--Eilenberg cohomology of $\fn$ is concentrated in degree zero. The associated multiplet to $\mathrm H^0(\fn)$ is then simply quasi-isomorphic to the de Rham complex on the supertranslation group $\Omega^\bullet(N)$.

In contrast, for cases with higher Lie algebra cohomology, there is no direct quasi-isomorphism between the multiplets associated to $C^\bullet(\fn)$ and $\mathrm H^\bullet(\fn)$. However---as advertised in Section~\ref{sec: components}---we can use the derived formalism to study the multiplets associated to the individual Lie algebra cohomology groups and their relationships. As discussed, there is a quasi-isomorphism identifying the multiplet associated to $C^\bullet(\fn)$ with the de Rham complex on spacetime,
\begin{equation*}
A^\bullet(C^\bullet(\fn)) \simeq \Omega^\bullet(N) .
\end{equation*}
On the other hand, we can first take cohomology with respect to the Chevalley--Eilenberg differential. Let us filter the complex $A^\bullet(C^\bullet(\fn))$ by the total degree on the ring $C^\bullet(\fn)$: the internal $\ZZ$-degree plus the $\RR$-weight: the Chevalley--Eilenberg cohomology has degree one for this filtration, and the remaining piece of the differential has degree zero. The homotopy transfer theorem leads to an equivalence
\begin{equation*}
A^\bullet(C^\bullet(\fn)) \simeq \bigg( \bigoplus_k A^\bullet\big(\mathrm H^k(\fn)\big) , \cD' \bigg) ,
\end{equation*}
where the differential $\cD'$ on the right hand side is induced by homotopy transfer, and consists of a sum of terms
\begin{equation*}
\cD'_{j,k} \colon \ A^\bullet\big(\mathrm H^k(\fn)\big) \to A^\bullet\big(\mathrm H^{k-j}(\fn)\big)
\end{equation*}
for $j \ge 1$. An alternative way of understanding these differentials is to consider higher differentials in the spectral sequence associated to the filtration discussed above, by total degree on~$C^\bullet(\fn)$.

For example, let us consider the case $j=1$. The process we have just discussed gives rise to differential operators between the associated component field multiplets:
\begin{equation*}
\nabla \colon \ \mu A^\bullet\big(\mathrm H^k(\fn)\big) \longrightarrow \mu A^\bullet\big(\mathrm H^{k-1}(\fn)\big) .
\end{equation*}

As we will see in an example momentarily, these operators can intuitively be thought of in the case $k=0$ by viewing $\mu A^\bullet\big(\mathrm H^{-1}\big)$ as the field strength multiplet of $\mu A^\bullet\big(\mathrm H^0\big)$, with $\nabla$ acting as the field strength or curvature map.
In particular, as discussed in~\cite{perspectives}, the multiplet associated to~$H^0$ always contains a summand corresponding to a $p$-form abelian gauge field for some $p$. Since the deformation we describe deforms the sum of the minimal multiplets $\mu A^\bu\big(\mathrm H^k(\fn)\big)$ to (an object quasi-isomorphic to) the de Rham complex on spacetime, one expects that the operator mapping $\mu A^\bu\big(\mathrm H^0\big)$ to $\mu A^\bu\big(\mathrm H^{-1}\big)$ will include a de Rham differential that carries the $p$-form to a~$(p+1)$-form ``field strength'' component field in the multiplet $\mu A^\bu\big(\mathrm H^{-1}\big)$. In examples, we will see that this is the case.

We expect that our operators should provide a general coordinate-free description of operators of the type ``$R^a$'', considered by Cederwall in constructing pure spinor superfield actions; see~\cite{Ced-ops}. These operators appear, for example, in the discussion of interacting eleven-dimensional supergravity in~\cite{Ced-11d, Ced-towards}, and for the Dirac--Born--Infeld action in~\cite{CederwallKarlsson}.

\begin{eg}
	Consider $\cN=1$ supersymmetry in three dimensions. Let $\fp$ denote the $\cN=1$ super Poincar\'e algebra in three dimensions, and let $\fg = \fp \oplus \RR$, where $\RR$ acts on $\fn$ by rescaling, with weight one on the odd piece and with weight two on the even piece. We will use the $\RR$ action to lift the $\ZZ \times \ZZ/2\ZZ$-grading on our multiplets. We denote the associated super translation algebra by $\fn= \fg_{>0}$ and the super translation group by $N$.
	Recall that the supertranslation algebra is of the form
	\begin{equation*}
	\fn = S[-1] \oplus V[-2],
	\end{equation*}
	where $S$ is the two-dimensional spin representation of $\mathrm{Spin}(3) \cong \mathrm{SU}(2)$, and $V$ is the three-dimensional vector representation of $\mathrm{Spin}(3)$, with the bracket being induced from the equivariant isomorphism $\Sym^2(S) \cong V$. We can choose bases for these representations to obtain generators and relations for the Chevalley--Eilenberg complex as follows:
	\begin{equation*}
	C^\bullet(\fn) = \CC\big[\lambda^\alpha , v^\mu\big], \qquad \d_{\mathrm{CE}}v^1 = \big(\lambda^1\big)^2 ,\qquad \d_{\mathrm{CE}}v^2 = \lambda^1 \lambda^2, \qquad \d_{\mathrm{CE}}v^3 = \big(\lambda^2\big)^2 .
	\end{equation*}
	It will sometimes be convenient to write the differential more compactly as
	\begin{equation*}
	\d_{\mathrm{CE}}v^{(\alpha \beta)} = \lambda^{(\alpha} \lambda^{\beta)} .
	\end{equation*}
	The Lie algebra cohomology is easily computed.
	\begin{equation*}
	\mathrm H^k = \begin{cases}
	R/I & \text{if } k=0, \\
	\big(\big(\lambda^2 v^1 - \lambda^1 v^2\big) \CC\big[\lambda^1\big] \oplus \big(\lambda^2 v^2 - \lambda^1 v^3\big) \CC\big[\lambda^2\big] \big)/ I_{-1} & \text{if } k=-1, \\
	0 & \text{otherwise.}
	\end{cases}
	\end{equation*}
	The ideal $I_{-1}$ appearing in the case where $k=-1$ is spanned by $\lambda^2 \big(\lambda^2 v^1 - \lambda^1 v^2\big) - \lambda^1\big(\lambda^2 v^2 - \lambda^1 v^3\big)$. As explained in~\cite{perspectives}, $\mu A^\bullet\big(\mathrm H^0\big)$ is identified with the 3d $\cN=1$ gauge multiplet and $\mu A^\bullet\big(\mathrm H^{-1}\big)$ with the corresponding antifield multiplet. In the array notation introduced in Section~\ref{ssec:gradings}, we can denote the direct sum of these multiplets as follows:
	\begin{equation*}
	\mu A^\bu\big(\mathrm H^0\big) \oplus \mu A^\bu\big(\mathrm H^{-1}\big) = \left[
	\begin{tikzcd}[column sep = 0.3cm , row sep = 0.3cm]
	\Omega^0 \arrow[dr ,"\d"] & & \\
	& \Omega^1 & S \\
	& S & \Omega^2 \arrow[dr , "\d"] \\
	& & & \Omega^3
	\end{tikzcd} \right].
	\end{equation*}
	This direct sum is not yet isomorphic to the de Rham complex, but we can clearly see how the differential can be deformed such that this is the case: one has to add an acyclic piece which cancels the two spin representations and an additional differential of order one forming the de Rham differential between $\Omega^1$ and $\Omega^2$.
	
	Let us now see how this deformation arises from the differential on $A^\bullet(C^\bullet(\fn))$ via homotopy transfer. We start by taking cohomology with respect to the Chevalley--Eilenberg differential~$\d_{\mathrm{CE}}$ and fix homotopy data
	\begin{equation*}
	\begin{tikzcd}
	\arrow[loop left]{l}{h}(A^{\bu} (C^\bu(\fn)) , \d_{\mathrm{CE}})\arrow[r, shift left, "p"] &\big( C^\infty(N) \otimes \mathrm H^\bullet(\fn) ,   0\big) \arrow[l, shift left, "i"]
	\end{tikzcd}\!\!,
	\end{equation*}
	such that we find induced differentials of the form $\cD'= p\cD i + p \cD h \cD i + \cdots$. By degree reasons, only the first and the second summand can contribute non-zero terms. Taking cohomology with respect to $\cD_0$ and transferring to the minimal multiplets we finally obtain the map
	\begin{equation*}
	\nabla \colon \ \mu A^\bullet\big(\mathrm H^0\big) \longrightarrow \mu A^\bullet\big(\mathrm H^{-1}\big) ,
	\end{equation*}
	which we view as a field strength map.
	Let us now study this map explicitly using our basis.
	Running the techniques described in~\cite{perspectives} (see also~\cite{KroLoCS} for an earlier account), one finds the following representatives for the component fields in $\mu A^\bullet\big(\mathrm H^0\big)$:
	\begin{equation*}
	\begin{bmatrix}
	1 & - & -\\
	- & \lambda^{(\alpha} \theta^{\beta)} & \lambda^\alpha \theta^2
	\end{bmatrix},
	\end{equation*}
	where internal degree is indicated by vertical position---concentrated in degrees zero and one---and the $\RR^\times$-weight minus the degree is indicated by horizontal position. For $\mu A^\bu\big(\mathrm H^{-1}\big)$ one finds
	\begin{equation*}
	\begin{bmatrix}
	\lambda_{\alpha} v^{(\alpha \beta)} & \lambda^{(\alpha} \theta^{\beta)} v^{\gamma \delta} & -\\
	- & - & \theta^2 \lambda_{\alpha} \lambda_\beta v^{(\alpha \beta)}
	\end{bmatrix} ,
	\end{equation*}
	where now the rows indicate degree 2 and 3. We can fix a homotopy $h$ for $\d_{\mathrm{CE}}$ by
	\begin{equation*}
	h\big(\lambda^\alpha \lambda^\beta\big) = v^{(\alpha \beta)} .
	\end{equation*}
	There will be terms in $\nabla$ of the form
	\begin{equation*}
	p \lambda \frac{\partial}{\partial \theta} h \lambda \frac{\partial}{\partial \theta} i .
	\end{equation*}
	Acting on representatives of the fermions in the multiplet $\mu A^\bullet\big(\mathrm H^0\big)$, we find
	\begin{equation*}
	\lambda^\alpha \theta^2 \mapsto \lambda^\alpha \lambda^\beta \theta_\beta \mapsto v^{(\alpha \beta)} \theta_{\beta} \mapsto \lambda_{\beta} v^{(\alpha \beta)}.
	\end{equation*}
	This is a map between $\mu A^\bu\big(\mathrm H^0\big)$ and $\mu A^\bu\big(\mathrm H^{-1}\big)$, which induces an acyclic deformation on their direct sum.
	
	In addition, $\nabla$ contains terms of order one given by
	\begin{equation*}
	p v \frac{\partial}{\partial x} i .
	\end{equation*}
	Investigating the representatives, it is clear that this acts as a de Rham differential between the one-form in $\mu A^\bu\big(\mathrm H^0\big)$ and the two-form in $\mu A^\bu\big(\mathrm H^{-1}\big)$. This maps the one-form gauge field present in $\mu A^\bu\big(\mathrm H^0\big)$ to its field strength, which is part of $\mu A^\bu\big(\mathrm H^{-1}\big)$.
	
\end{eg}

\subsection{Antifield multiplets in four dimensions} \label{sec: vector anti}

Let us discuss an example of a multiplet which is of physical relevance, but is not in the essential image of the functor $A^\bullet_{R/I}$. For four-dimensional $\cN=1$ supersymmetry, $\mu A^\bullet(R/I)$ can be identified with the vector multiplet. The corresponding antifield multiplet, however, cannot be constructed in the underived pure spinor superfield formalism. This failure is linked to the fact that the relevant nilpotence variety is not Cohen--Macaulay~\cite{perspectives}.

Here we describe the antifield multiplet using component fields first and then calulate the derived $\fn$-invariants. As expected we find that the cohomology is not concentrated in a single degree. We begin by recalling the structure of the supertranslation algebra in four dimensions.

\begin{dfn}
	The four-dimensional $\mathcal N=1$ supertranslation algebra is the super Lie algebra with underlying $\ZZ/2\ZZ$-graded $\mathrm{Spin}(4)$-module
	\begin{equation*}\fn = \RR^4 \oplus \Pi(S_+ \oplus S_-),\end{equation*}
	with Lie bracket given by the canonical isomorphism $S_+ \otimes S_- \to \RR^4$ of $\mathrm{Spin}(4)$-representations.
\end{dfn}

\begin{eg}
	Let $(E, D, \rho)$ denote the antifield multiplet to the vector multiplet on $\RR^4$. It is concentrated in $\RR$-weight 0 to 4 and cohomological degree 0 and 1, and can be concretely described in array notation as
	\begin{equation} \label{eq: anti vec}
	\mu A^\bu(R/I)^* = \left[
	\begin{tikzcd}[row sep=0.7cm, column sep=0.7cm]
	\Omega^0\big(\RR^4\big) & \Gamma\big(\mathbb{R}^4,S_+ \oplus S_-\big) & \Omega^1\big(\mathbb{R}^4\big) \arrow[dr, "\star \d \star"] \\
	& & &\Omega^0\big(\mathbb{R}^4\big)
	\end{tikzcd} \right] .
	\end{equation}
	We can also describe the $L_\infty$ action of the supertranslation algebra $\fn$ concretely. We will only need the action on constant sections. There, odd elements act by the following formula:
	\begin{alignat*}{4}
	&\rho^{(1)}_{\mathrm{constants}}(Q) \colon  \quad&& S_+ \oplus S_- \longrightarrow \Omega^0 ,\qquad &&\big(\psi^\vee , \bar{\psi}^\vee\big) \mapsto Q_+ \wedge \psi^\vee + Q_- \wedge\bar{\psi}^\vee,& \\
	 &&&\Omega^1 \longrightarrow S_+ \oplus S_- , \qquad&& A^\vee \mapsto Q_+ \wedge A^\vee + Q_- \wedge A^\vee ,&\\
	&\rho^{(2)}_{\mathrm{constants}}(Q_1,Q_2) \colon  \quad&& \Omega^0 \longrightarrow \Omega^1 ,\qquad&& c^\vee \mapsto \Gamma(Q_1,Q_2) \otimes c^\vee,&
	\end{alignat*}
	where we have denoted the fields by $\big(\psi^\vee, \bar{\psi}^\vee\big) \in \Gamma\big(\mathbb{R}^4,S_+ \oplus S_-\big)$, $A^\vee \in \Omega^1\big(\mathbb{R}^4\big)$, and $c^\vee \in \Omega^0\big(\mathbb{R}^4\big)$ in degree one. We have similarly denoted the positive and negative helicity summands of $Q \in S_+ \oplus S_-$ by $Q_+$ and $Q_-$ respectively.
\end{eg}

Let us write $E_0$ for the fiber of $E$ over $0 \in \RR^4$ (not to be confused with the summand of degree zero), so concretely
\begin{equation*} 
E_0 = \left[
\begin{tikzcd}[row sep=0cm, column sep=0.25cm]
\RR & S_+ \oplus S_- & \RR^4 \\
& & &\RR
\end{tikzcd} \right] .
\end{equation*}
We can use this to describe the Chevalley--Eilenberg complex of $\fn$ with coefficients in our multiplet, using the $L_\infty$ action of $\fn$ on the sheaf $\mathcal E$ of sections of $E$. We find the following:

\begin{prop}
	\begin{equation*}
	C^\bullet(\fn, \mathcal E) \cong (E_0 \otimes \Sym(S_+ \oplus S_-), \d_\rho),
	\end{equation*}
	where $\d_\rho$ is generated over $\Sym(S_+ \oplus S_-)$ by the sum of the terms
	\begin{align*}
	&\rho^{(1)}|_{\mathrm{constants}} \colon \ E_0 \otimes (S_+ \oplus S_-) \to E_0, \\
	&\rho^{(2)}|_{\mathrm{constants}} \colon \ E_0 \otimes \Sym^2(S_+ \oplus S_-) \to E_0
	\end{align*}
	obtained by restricting the $L_\infty$ action $\rho$ to constant sections of $\mathcal E$.
\end{prop}

\begin{proof}
	Note that
	\begin{equation*}
	C^\bullet(\fn) \cong \big(\Sym\big(\big(\RR^4\big)^*[1]\big) \otimes \Sym(S_+ \oplus S_-), \d_{\mathrm{CE}}\big),
	\end{equation*}
	where we have identified $S_\pm$ with its dual using its canonical inner product. We obtain the given description by taking the cohomology by the operator dual to the action of the algebra of translations on $\mathcal E$; the result is quasi-isomorphic to $C^\bullet(\fn, \mathcal E)$, no additional homotopical correction terms appear.
\end{proof}

Let us now discuss the cohomology of the Chevalley--Eilenberg complex.

\begin{prop}
	We have an isomorphism
	\begin{equation*}
	\mathrm H^{\bullet}(\fn, \cE) \cong \CC \oplus (\Sym(S_+) \oplus \Sym(S_-))[-1].
	\end{equation*}
\end{prop}

\begin{proof}
	This is a straightforward calculation using the description that we have just given. In the weight zero term of $E_0$, all elements are $\d_\rho$-closed, but all such elements other than constants are also $\d_\rho$-exact. In weight 1 in $E_0$, the $\d_\rho$-closed elements are generated over $1 \otimes \Sym(S_+ \oplus S_-)$ by $\wedge^2 S_+ \oplus \wedge^2 S_- \subseteq (S_+ \oplus S_-) \otimes (S_+ \oplus S_-)$. When we quotient by $\d_\rho$-exact elements we are left with $\bigl(\wedge^2 S_+ \otimes \Sym(S_-)\bigr) \oplus \bigl(\wedge^2 S_- \otimes \Sym(S_+)\bigr)$. In weight 2 in $E_0$ the closed and exact elements coincide, and in weight $>2$ in $E_0$ there are no $\d_\rho$-closed elements.
\end{proof}

From our general results, we know abstractly that $A^\bullet(C^\bullet(\fn, \cE))$ is dual to the vector multiplet. It is instructive to calculate the component field formulation for this multiplet explicitly using the recipe presented in Section~\ref{sec: components} in order to see how the is related to the corresponding calculation in the underived pure spinor superfield formalism.

Our calculation will follow the same outline as the calculation we performed in the previous section. We will first study the multiplet associated to each of the two cohomology groups of~$C^\bullet(\fn, \mathcal E)$ in isolation, then we will compute the additional differential in $A^\bullet(\mathrm H^\bullet(\fn, \mathcal E))$ relating these two individual terms, obtained by homotopy transfer.

First, recall that
\begin{equation*}
A^\bullet(\CC) \simeq C^\infty(N)
\simeq \Gamma\big(\RR^4, \Sym(S_+[1] \oplus S_-[1])\big).
\end{equation*}

For the non-trivial summand of the cohomology, we can compute that
\begin{equation*}
A^\bullet(\Sym(S_+)) \simeq \big(\Gamma\big(\RR^4, \Sym(S_+) \oplus \Sym(S_+[1] \oplus S_-[1])\big), \cD\big)
\simeq \Gamma\big(\RR^4, \Sym(S_-[1])\big),
\end{equation*}
where the differential $\cD$ is generated by the degree one isomorphism $S_+[1] \to S_+$. Similarly
\begin{equation*}
A^\bullet(\Sym(S_+)) \simeq \Gamma\big(\RR^4, \Sym(S_+[1])\big).
\end{equation*}
So altogether, when we compute $A^\bullet(\mathrm H^\bullet(\fn, \mathcal E))$, we obtain a multiplet with the following Betti numbers:
\begin{equation*}
\begin{bmatrix}
1 & 4 & 6 & 4 & 1 \\
- & 2 & 4 & 2 &-
\end{bmatrix} .
\end{equation*}
Now, let us compute the correction terms that allow us to obtain $A^\bullet(C^\bullet(\fn, \mathcal E))$ in full. As discussed in the previous section, Section~\ref{multiplet_from_cohomology_section}, there is an additional differential coming from homotopy transfer. We will compute this differential in coordinates.

Running the procedure described in~\cite{perspectives}, we find the following local coordinates for our multiplet:
\begin{equation*} 
\begin{bmatrix}
1 & \big(\theta_\alpha,\bar{\theta}_{\dot \alpha}\big) & \big(\theta^2 , \theta_\alpha \bar{\theta}_{\dot \alpha}, \bar{\theta}^2\big) & \big(\theta^2 \theta_\alpha , \bar{\theta}^2 \bar{\theta}_{\dot \alpha}\big) & \theta^2 \bar{\theta^2} \\
- & \big(\lambda s , \bar{\lambda} \bar{s}\big) & \big(\lambda s \theta_\alpha , \bar{\lambda} \bar{s} \bar{\theta}_{\dot \alpha}\big) & \big(\lambda s \theta^2 , \bar{\lambda} \bar{s} \bar{\theta}^2\big) & -
\end{bmatrix}.
\end{equation*}
Let us unpack the notation. Recall that we identified $C^\bullet(\fn, \mathcal E)$ with $(E_0 \otimes \Sym(S_+ \oplus S_-), \d_\rho)$. We use $\big\{s^\alpha, \bar s^{\dot \alpha}\big\}$ for a basis of $S_+ \oplus S_-$ in the first factor, and $\{\lambda^\alpha, \bar \lambda^{\dot \alpha}\}$ for a basis of $S_+ \oplus S_-$ in the second factor, and, e.g., $\lambda s$ is the image of 1 under the canonical map $1 \to S_+ \otimes S_+$ dual to the scalar valued pairing. We then use $\big\{\theta_\alpha, \bar \theta_{\dot \alpha}\big\}$ for the linear odd coordinate functions $S_+^* \oplus S_-^* \subseteq C^\infty(N)$.

With this concrete basis in hand, let us now investigate the additional pieces in the differential coming from homotopy transfer. For this purpose, let us fix a homotopy $h$ for the differential~$\d_\rho$. We have
\begin{equation*}
\d_\rho (s_\alpha) = \lambda_\alpha ,
\end{equation*}
and therefore we can set
\begin{equation*}
h(\lambda_\alpha) = s_\alpha
\end{equation*}
and similarly for $\bar{\lambda}$ and $\bar{s}$. There are two nontrivial terms contributing to the transferred differential
\begin{equation*}
\cD_0 h \cD_0 \big(\theta^2\big) = \lambda s
\end{equation*}
and again similarly for the complex conjugates. We see that this induces a differential that cancels some of the representative basis elements in pairs. The remaining representatives are
\begin{equation*} 
\begin{bmatrix}
1 & \big(\theta,\bar{\theta}\big) & \theta \bar{\theta} & -\\
- & - & - & \lambda s \theta^2 + \bar{\lambda} \bar{s} \bar{\theta}^2
\end{bmatrix}.
\end{equation*}
It is immediate to see that these representatives span the $\mathrm{Spin}(4)$ representations occurring in~\eqref{eq: anti vec}. In addition there is a differential of order one described by
\begin{equation*}
\cD_0 h \cD_1 \big(\theta^\alpha \bar{\theta}^{\dot{\beta}} A_{\alpha \dot{\beta}} \big) = \big(\lambda s \theta^2 + \bar{\lambda} \bar{s} \bar{\theta}^2\big) \partial^\mu A_\mu .
\end{equation*}
We thus recover the anticipated description of the vector antifield multiplet.

\subsection{The chiral multiplet revisited}
Let us discuss one further example in the derived formalism, which will illustrate the relation between on- and off-shell multiplets within the formalism, i.e., the appearance of non-trivial $L_\infty$ actions of the supersymmetry algebra. We will again work specifically with four-dimensional~$\mathcal N=1$ supersymmetry.

In~\cite{perspectives}, it was demonstrated that the minimal multiplet $\mu A^\bullet_{R/I}(\Gamma)$ associated to the module~$\Gamma = \Sym(S_+)$ is equivalent to the BRST version of the chiral multiplet. Constructing the associated cotangent theory yields the standard off-shell BV theory of the chiral multiplet whose component fields include a scalar field $\phi$, a chiral spinor field $\psi$, and an auxiliary scalar field $F$, as well as their associated complex conjugates and antifields. Of course, one can integrate out the auxiliary field $F$ and obtain an equivalent BV theory, but the supersymmetry algebra action is now no longer strict. This is referred to as the on-shell formulation of the chiral multiplet. This is discussed in the $L_\infty$-module language in \cite{SWchar}, but the idea is much older, for example, the related example of an $\cN=2$ hypermultiplet is discussed in the on-shell language in~\cite{BaulieuSusy}.

Let us again start from the component field formulation for these multiplets and compute their derived $\fn$-invariants. Plugging in this module in the derived pure spinor formalism we find that the minimal multiplet $\mu A^\bullet(C^\bullet(\fn,\cE))$ can be explicitly identified with the on-shell formulation for the chiral multiplet described above. The off-shell formulation including the auxiliary field is given by a quasi-isomorphic non-minimal multiplet.

The component fields of the chiral multiplet in the on-shell formulation take the following form.
\begin{equation*} \label{eq:chira on}
E= \left[\begin{tikzcd}[row sep=0.7cm, column sep=0.7cm]
\Omega^0 \otimes \CC^2 \arrow[drrr, "\star \d \star \d" ' near start] & \Omega^0 \otimes (S_+ \oplus S_-) \arrow[dr, crossing over, "\slashed{\partial}" ]\\
& & \Omega^0 \otimes (S_+ \oplus S_-) & \Omega^0 \otimes \CC^2
\end{tikzcd}\right].
\end{equation*}
In order to describe the $\fn$ action, we will denote the component fields by $\big(\phi, \bar{\phi}\big)$ in degree zero, weight zero and $\big(\psi,\bar{\psi}\big)$ in degree zero, weight one and their respective antifields in degree one by $\big(\phi^\vee, \bar{\phi}^\vee\big)$ and $\big(\psi^\vee,\bar{\psi}^\vee\big)$. The odd elements of $\fn$ act in the following way.
\begin{alignat*}{3}
& \rho^{(1)}(Q) \colon \quad && S_+ \oplus S_- \longrightarrow \Omega^0 \otimes \CC^2 ,\qquad \big(\psi , \bar{\psi}\big) \mapsto p_+(Q) \wedge \psi + p_-(Q) \wedge \bar{\psi},&\\
&&& \Omega^0 \otimes\CC^2 \longrightarrow S_+ \oplus S_- ,\qquad \big(\phi, \bar{\phi}\big) \mapsto Q \wedge \slashed{\partial} \big(\phi + \bar{\phi}\big), &\\
&\rho^{(2)}(Q,Q) \colon \quad&& S_+ \oplus S_- \longrightarrow S_+ \oplus S_- ,&\\
&&& \big(\psi^+, \bar{\psi}^+\big) \mapsto p_-(Q) \otimes p_+(Q) \wedge \psi^\vee + p_+(Q) \otimes p_-(Q) \wedge \bar{\psi}^\vee.&
\end{alignat*}
With this description we can compute the cohomology of the Chevalley--Eilenberg complex. We will use the following notation. Let $(e, \bar{e})$ and $(s,\bar{s})$ denote bases of $\CC^2$ and $S_+ \oplus S_-$ respectively. As in the previous section, let us use $\lambda_\alpha$, $\bar \lambda_{\dot \alpha}$ to denote even basis elements for $C^\bullet(\fn)$, and write~$\lambda s$, $\bar \lambda \bar s$ to indicate the images of the canonical maps $\RR \to S_\pm \otimes S_\pm^*$. One finds the following description for the Chevalley--Eilenberg cohomology.
\begin{equation*}
\mathrm H^\bullet(\fn, \cE) \cong \left(\Sym(S_-) e \oplus \Sym(S_+) \bar e\right) + \left(\Sym(S_+) \lambda s \oplus \Sym(S_-) \bar \lambda \bar s\right)[-1].
\end{equation*}
Again, we see that the cohomology is not concentrated in a single degree.

We can calculate the multiplet associated to this module following a method similar to the previous section, beginning with the multiplets associated to the summands of the cohomology described above, and then computing the correction terms associated to homotopy transfer. We obtain a quasi-isomorphic multiplet
\begin{equation*}
\left( \bigoplus_k A^\bullet\big(\mathrm H^k(\fn, \cE)\big) , \cD' , \rho' \right) ,
\end{equation*}
where the new differential $\cD'$ contains terms induced via homotopy transfer. Individually, both $\mu A^\bullet\big(\mathrm H^0(\fn, \cE)\big)$ and $\mu A^\bullet\big(\mathrm H^1(\fn, \cE)\big)$ contain the field content of a chiral and an antichiral BRST multiplet. We have the following Betti numbers for the sum of the multiplets induced from the individual cohomology groups:
\begin{equation*}
\begin{bmatrix}
2 & 4 & 2 & - \\
- & 2 & 4 & 2
\end{bmatrix}.
\end{equation*}
There are explicit elements representing the cohomology which take the forrm
\begin{equation*}
\begin{bmatrix}
\big(e,\bar{e}\big) & \big(\theta_\alpha e ,\bar{\theta}_{\dot \alpha} \bar{e}\big) & \big(\theta^2 e, \bar{\theta}^2 ,\bar{e}\big) & -\\
- & \big(\lambda s , \bar{\lambda} \bar{s}\big) & \big(\lambda s \bar{\theta}_{\dot \alpha} , \bar{\lambda} \bar{s} \theta\big)_\alpha & \big(\lambda s \bar{\theta}^2 , \bar{\lambda} \bar{s} \theta^2\big)
\end{bmatrix} .
\end{equation*}

Now let's take a look at the induced differentials under homotopy transfer. For an explicit homotopy $h$ we can choose
\begin{equation*}
h(\lambda_\alpha e) = s_\alpha, \qquad h \big(\bar{\lambda}_{\dot \alpha} \bar{e}\big) = \bar{s}_{\dot \alpha} .
\end{equation*}
Applying this to the representatives, we find
\begin{gather*}
p \cD_0 h \cD_0 i (F) = p \cD_0 h \cD_0 \big(F \theta^2 e\big) = p((\lambda s) F) = F, \\
p \cD_1 h \cD_0 i (\psi) = p \cD_1 h \cD_0 (\psi \theta e) = p\big((\lambda s) \bar{\theta} \slashed{\partial} \psi\big) = \slashed{\partial} \psi,\\
p \cD_1 h \cD_1 i (\phi) = p \cD_1 h \cD_1 (\phi e) = p\big((\lambda s) \bar{\theta}^2 \partial^2 \phi\big) = \partial^2 \phi .
\end{gather*}
Similar results hold for the complex conjugate fields. Thus, we find that there are induced differentials of order zero, one and two. With respect to the weight grading, these operators alter the weight grading by minus one, one and three respectively. We can summarize the multiplet in the following diagram.
\begin{equation*}
\left[
\begin{tikzcd}[row sep=0.7cm, column sep=0.7cm]
\Omega^0 \otimes \CC^2 \arrow[drrr, "\star \d \star \d" ' near start] & \Omega^0 \otimes (S_+ \oplus S_-) \arrow[dr, "\slashed{\partial}" , near start ] &\Omega^0 \otimes \CC^2 \arrow[dl , "\mathrm{id}" ' near start, crossing over]\\
& \Omega^0 \otimes \CC^2 & \Omega^0 \otimes (S_+ \oplus S_-) & \Omega^0 \otimes \CC^2
\end{tikzcd} \right].
\end{equation*}
This is precisely the off-shell BV model for the chiral multiplet. Taking cohomology with respect to the acyclic differential we obtain the minimal multiplet $\mu A^\bullet(C^\bullet(\fn,\cE))$ which is precisely the original on-shell formulation that we started with.

\subsection*{Acknowledgements}
We would like to give special thanks to R.~Eager, J.~Walcher, and B.~R.~Williams for conversations and collaboration on related projects. We also gratefully acknowledge conversations with I.~Brunner, M.~Cederwall, I.~Contreras, O.~Gwilliam, J.~Huerta, J.~Palmkvist, and S.~Noja. In addition, we would like to thank the anonymous referees for useful suggestions. This work is funded by the Deutsche Forschungsgemeinschaft (DFG, German Research Foundation) under Germany’s Excellence Strategy EXC 2181/1 -- 390900948 (the Heidelberg STRUCTURES Excellence Cluster). I.S.\ is supported by the Free State of Bavaria.

\pdfbookmark[1]{References}{ref}
\LastPageEnding

\end{document}